\documentclass[11pt]{article}
\usepackage[utf8]{inputenc}
\usepackage{amsmath}
\usepackage{amssymb}
\usepackage{amsthm}
\usepackage[boxed,lined,linesnumbered, ruled]{algorithm2e}
\usepackage{enumerate}
\usepackage[dvipsnames]{xcolor}
\usepackage{comment}
\usepackage{nicefrac}
\usepackage{fullpage}
\usepackage{graphics}
\usepackage{graphicx}
\usepackage{wrapfig}
\usepackage{hyperref}
\usepackage[T1]{fontenc}
\usepackage{thmtools,thm-restate}
\usepackage{adjustbox}
\usepackage{bm}
\usepackage{soul}
\usepackage{amsmath}
\usepackage[framemethod=tikz]{mdframed}
\mdfsetup{frametitlealignment=\centering}

\usepackage[letterpaper,margin=1in]{geometry}

\usepackage{graphicx}
\usepackage{subcaption} 
\usepackage{tikz}
\usetikzlibrary{shapes.geometric}
\usetikzlibrary{arrows.meta}
\usetikzlibrary{positioning}
\usetikzlibrary{shapes.multipart}
\usepackage[multiple]{footmisc}

\newtheorem{theorem}{Theorem}[section]
\newtheorem{lemma}[theorem]{Lemma}
\newtheorem{corollary}[theorem]{Corollary}
\newtheorem{definition}[theorem]{Definition}
\newtheorem{claim}[theorem]{Claim}

\newtheorem{observation}[theorem]{Observation}

\newtheorem{fact}[theorem]{Fact}

\newcommand{\scatterdim}[1][$\epsilon$]{{{#1}-scatter dimension}\xspace}
\newcommand{\ScatterDim}[1][$\epsilon$]{{{#1}-Scatter Dimension}\xspace}
\newcommand{\algscatterdim}[1][$\epsilon$]{{algorithmic \scatterdim[#1]}\xspace}
\newcommand{\AlgScatterDim}[1][$\epsilon$]{{Algorithmic \ScatterDim[#1]}\xspace}
\newcommand{\scattering}[1][$\epsilon$]{{{#1}-scattering}\xspace}

\newcommand{\algscattering}[1][$(\cC_{\cM},\epsilon)$]{{#1}-scattering\xspace}
\newcommand{\opt}{\ensuremath{\textsf{OPT}}}

\newcommand{\poly}{\ensuremath{\textsf{poly}}}
\newcommand{\polylog}{\ensuremath{\textsf{polylog}}}

\newcommand{\ball}{{\sf ball}}

\newcommand{\km}{\textsc{$k$-Median}\xspace}
\newcommand{\okm}{\textsc{Ordered $k$-Median}\xspace}
\newcommand{\wkm}{\textsc{Weighted $k$-Median}\xspace}
\newcommand{\kzc}{\textsc{$(k,z)$-Clustering}\xspace}
\newcommand{\knc}{\textsc{Norm $k$-Clustering}\xspace}
\newcommand{\kmns}{\textsc{$k$-Means}\xspace}
\newcommand{\kc}{\textsc{$k$-Center}\xspace}
\newcommand{\lcentrum}{\textsc{$\ell$-Centrum}\xspace}
\newcommand{\rkm}{\textsc{Robust $k$-Median}\xspace}
\newcommand{\sfm}{\textsc{Socially Fair $k$-Median}\xspace}

\newcommand{\pqfc}{\textsc{$(z,q)$-Fair Clustering}\xspace}

\newcommand{\pkc}[1][$k$]{\textsc{Priority {#1}-Center}\xspace}
\newcommand{\wkc}{\textsc{Weighted $k$-Center}\xspace}
\newcommand{\pc}{\textsc{Ball Intersection}\xspace}

\newcommand{\of}{\textsf{OPT}_{\textnormal{f}}\xspace}

\newcommand{\vc}[1]{\bm{#1}\xspace}
\newcommand{\vcx}{\bm{x}\xspace}
\newcommand{\vcg}{\bm{g}\xspace}
\newcommand{\vcw}{\bm{w}\xspace}

\newcommand{\vcd}{\bm{\delta}\xspace}

\newcommand{\defproblemout}[3]{
  \vspace{1mm}
\noindent\fbox{
  \begin{minipage}{0.96\textwidth}
  \begin{tabular*}{\textwidth}{@{\extracolsep{\fill}}lr} #1 \\ \end{tabular*}
  {\bf{Input:}} #2  \\
  {\bf{Output:}} #3
  \end{minipage}
  }
  \vspace{1mm}
}

\newcommand{\spider}{spider}

\newcommand{\tw}{\textnormal{\texttt{tw}}}
\newcommand{\OO}{\mathcal{O}}

\newcommand{\cA}{\ensuremath{\mathcal{A}}\xspace}
\newcommand{\cC}{\ensuremath{\mathcal{C}}\xspace}

\newcommand{\cI}{\ensuremath{\mathcal{I}}\xspace}

\newcommand{\cM}{\ensuremath{\mathcal{M}}\xspace}

\newcommand{\bN}{\ensuremath{\mathbb{N}}\xspace}

\newcommand{\bR}{\ensuremath{\mathbb{R}}\xspace}
\newcommand{\bRd}{\ensuremath{\mathbb{R}^d}\xspace}
\newcommand{\bRn}{\ensuremath{\mathbb{R}^n_{\ge 0}}\xspace}
\newcommand{\bRp}{\ensuremath{\mathbb{R}^{P}_{\ge 0}}\xspace}
\newcommand{\bRna}{\ensuremath{\mathbb{R}^n}\xspace}

\newcommand{\bZ}{\ensuremath{\mathbb{Z}}\xspace}

\newcommand{\bigO}[1]{\ensuremath{\mathcal{O}\left(#1\right)}}

\newcommand{\prob}[1]{\ensuremath{\text{\bf Pr}\left [#1\right]}\xspace}

\newcommand{\grad}{\ensuremath{\partial f}\xspace}
\newcommand{\gradapx}{\ensuremath{\partial_{\epsilon} f}\xspace}

\def\DEBUG{true}

\ifdefined\DEBUG
\newcommand{\attention}[1]{\textcolor{red}{** #1 **}} 

\def\rem#1{\marginpar{\raggedright\scriptsize #1}}
\newcommand{\roor}[1]{\rem{\textcolor{BlueViolet}{$\bullet$ R: #1}}}
\newcommand{\jbyr}[1]{\rem{\textcolor{RedViolet}{$\bullet$ Jarek: #1}}}
\newcommand{\pryr}[1]{\rem{\textcolor{Orchid}{$\bullet$ P: #1}}}
\newcommand{\amtr}[1]{\rem{\textcolor{BurntOrange}{$\bullet$ A: #1}}}
\newcommand{\danr}[1]{\rem{\textcolor{Cyan}{$\bullet$ D: #1}}}
\newcommand{\kamr}[1]{\rem{\textcolor{WildStrawberry}{$\bullet$ K: #1}}}
\newcommand{\joer}[1]{\rem{\textcolor{Green}{$\bullet$ J: #1}}}
\newcommand{\roohani}[1]{\rem{\textcolor{BlueViolet}{$\bullet$ R: #1}}}
\newcommand{\faab}[1]{\rem{\textcolor{Dandelion}{$\bullet$ F: #1}}}
\newcommand{\sandy}[1]{\rem{\textcolor{violet}{$\bullet$ Sa: #1}}}
\else
\newcommand{\attention}[1]{}

\newcommand{\ftmr}[1]{}
\newcommand{\faab}[1]{}
\newcommand{\sanr}[1]{}
\newcommand{\jbyr}[1]{}
\newcommand{\pryr}[1]{}
\newcommand{\amtr}[1]{}
\newcommand{\kamr}[1]{}
\newcommand{\joer}[1]{}
\newcommand{\roor}[1]{}
\newcommand{\danr}[1]{}
\newcommand{\sandy}[1]{}
\newcommand{\roohani}[1]{}
\fi

\hypersetup{pdfinfo={
  Title={Parameterized Approximation Schemes for Clustering with General Norm Objectives},
  Author={Fateme Abbasi, Sandip Banerjee, Jarosław Byrka, Parinya Chalermsook, Ameet Gadekar, Kamyar Khodamoradi, D\'{a}niel Marx,\\
  Roohani Sharma, Joachim Spoerhase},
  Keywords={}
}}

\title{Parameterized Approximation Schemes\\ for Clustering with General Norm Objectives}
\author{
Fateme Abbasi\footnotemark[1] \and 
Sandip Banerjee\footnotemark[1] \and 
Jaros\l{}aw Byrka\footnote{University of Wroc\l{}aw, Poland (\texttt{\{fateme.abbasi,jby\}@cs.uni.wroc.pl}, \texttt{sandip.ndp@gmail.com})}
\end{tabular} \endgraf
\begin{tabular}[t]{c}
Parinya Chalermsook\footnotemark[2] \and
Ameet Gadekar\footnote{Aalto University, Espoo, Finland (\texttt{\{parinya.chalermsook,ameet.gadekar\}@aalto.fi}} \and 
Kamyar Khodamoradi\footnote{University of British Columbia, Canada (\texttt{kamyar.khodamoradi@ubc.ca})} 
\end{tabular} \endgraf
\begin{tabular}[t]{c}
D\'{a}niel Marx\footnote{CISPA Helmholtz Center for Information Security, Saarbr\"{u}cken, Germany (\texttt{marx@cispa.de})} \and
Roohani Sharma\footnotemark[5] \and
Joachim Spoerhase\footnote{Max Planck Institute for Informatics, Saarbr\"{u}cken, Germany (\texttt{\{rsharma,jspoerha\}@mpi-inf.mpg.de})}$^{\;}$
\footnote{The University of Sheffield, United Kingdom (\texttt{j.spoerhase@sheffield.ac.uk})}
}
\date{}

\begin{document}

\begin{titlepage}
\maketitle

\begin{abstract}
\thispagestyle{empty}
This paper considers the well-studied algorithmic regime of  designing a $(1+\epsilon)$-approximation algorithm for a $k$-clustering problem that runs in time $f(k,\epsilon)\poly(n)$ (sometimes called an efficient parameterized approximation scheme or EPAS for short\footnotemark). 
Notable results of this kind include EPASes in the high-dimensional Euclidean setting for $k$-center [Bad\u{o}iu, Har-Peled, Indyk; STOC'02] as well as $k$-median, and $k$-means [Kumar, Sabharwal, Sen; J. ACM 2010].

Our main contribution is a clean and simple  EPAS that settles more than ten clustering problems (across multiple well-studied objectives as well as metric spaces) and unifies well-known EPASes.
 More specifically, our algorithm gives EPASes in the following settings: 

\begin{itemize}
    \item \textbf{Clustering objectives}: $k$-means, $k$-center, $k$-median, priority $k$-center, $\ell$-centrum, ordered $k$-median, socially fair $k$-median (aka robust $k$-median), or any other objective that can be formulated as minimizing a monotone (not necessarily symmetric!) norm of the distances of the points from the solution (generalizing the symmetric formulation introduced by Chakrabarty and Swamy [STOC'19]).
    \item \textbf{Metric spaces}: Continuous high-dimensional Euclidean spaces, metrics of bounded doubling dimension, bounded treewidth metrics, and planar metrics. 
\end{itemize}

Prior to our results, EPASes were only known for vanilla clustering objectives ($k$-means, $k$-median, and $k$-center) and  each such algorithm is tailored to work for the specific input metric and clustering objective (e.g., EPASes for $k$-means and $k$-center in ${\mathbb R}^d$ are conceptually very different). 
In contrast, our algorithmic framework is applicable to a wide range of well-studied objective functions in a uniform way, and is (almost) entirely oblivious to any specific metric structures and yet is able to effectively exploit those unknown structures. In particular, our algorithm is \emph{not} based on the (metric- and objective-specific) technique of coresets.

Key to our analysis is  a new concept that we call \emph{bounded \scatterdim}---an intrinsic complexity measure of a metric space that is a relaxation of the standard notion of \textit{bounded doubling dimension} (often used as a source of algorithmic tractability for geometric problems).  Our main technical result shows that two conditions are essentially sufficient for our algorithm to yield an EPAS on the input metric $M$ for any clustering objective: 
\begin{itemize}
\item[(i)] The objective is described by a monotone norm, and
\item[(ii)] the \scatterdim of $M$ is upper bounded by a function of $\epsilon$.
\end{itemize}
\end{abstract}

\end{titlepage}
\clearpage 

\tableofcontents
\thispagestyle{empty}
\clearpage

\pagenumbering{arabic}

\section{Introduction}\label{sec:introduction}

In the class of \emph{$k$-clustering} problems, we are interested in partitioning $n$ data  points into $k$ subsets called clusters, each of which is represented by a center. We aim at minimizing a certain objective based on the distances between the data points and their respective cluster centers.
This  is among the most fundamental optimization problems that arise routinely in both theory and practice and has received attention from various research communities, including optimization, data mining, machine learning, and computational geometry. 
Basic clustering problems such as \km, \kc, and \kmns have been researched for more than half a century and yet remain elusive from many perspectives of computation.

This paper considers a prominent and classic algorithmic regime for $k$-clustering in which one aims at designing  \textit{efficient parameterized approximation schemes} (EPAS)---a $(1+\epsilon)$ approximation algorithm that runs in time $h(k, \epsilon) \poly(n)$ for every $\epsilon >0$\footnotetext{Quick remarks: (i) An EPAS is not comparable to polynomial time approximation schemes (PTAS), (ii) before the term EPAS was invented some researchers call this type of approximation schemes a PTAS or simply an approximation scheme (in clustering, it is often assumed that $k$ is small)~\cite{kumar2010linear,feldman2007ptas}, and (iii) both EPAS and PTAS are implied by the existence of efficient polynomial time approximation schemes (EPTAS).}. 
In a general metric space, obtaining such an approximation scheme is impossible even for  basic clustering problems. Past research has therefore focused on designing algorithms that work in \emph{structured} metric spaces (such as planar graphs or Euclidean spaces). 
In the continuous high-dimensional Euclidean space, EPASes are arguably the ``fastest'' approximation scheme one can hope for~\cite{dasgupta08hardness-2-means,awasthi2015hardness}, so it is no surprise that research on EPASes for clustering problems has received a lot of attention in the past two decades~\cite{har2004coresets,matouvsek2000approximate,ostrovsky2013effectiveness,ding2020unified,bhattacharya2018faster,braverman2021coresets,cohen2019efficient}.\footnote{We remark that PTASes, which are incomparable to EPASes, do not exist for continuous $k$-means, $k$-median and $k$-center~\cite{cohen2019inapproximability,awasthi2015hardness}.}

This paper is inspired by the following meta-question: 

\begin{mdframed}[userdefinedwidth=16cm,align=center,backgroundcolor=gray!20]
For a given $k$-clustering objective and a (structured) metric space, does an EPAS exist? 
\end{mdframed}

Systematic understanding about this question has been seriously lacking. 
While affirmative answers for basic clustering problems such as \kc, \km, and \kmns in the continuous high-dimensional Euclidean space have been shown already two decades ago~\cite{badoiu-etal:approximate-clustering-coresets,kumar2010linear} (recently for structured graph metrics~\cite{baker2020coresets,katsikarelis2019structural,fox2019embedding,braverman2021coresets,cohen2019efficient}), we do not know of any such result for  more complex clustering objectives.

This paper makes substantial progress towards a complete understanding of  the above meta-question. In particular, we present a unified EPAS that works for a broad class of clustering objectives (encompassing almost all center-based clustering objectives ever considered by the algorithms community and some new ones that further generalize the existing problems) as well as diverse metric spaces, hence settling many well-studied standalone clustering problems as a by-product.\footnote{There are variants of clustering problems that enforce constraints on how points can be assigned to open centers (e.g., capacitated and diversity constraints). Our purpose is handling many center-based clustering objectives; handling a broad range of constraints (such as capacities) is beyond the scope of this paper.} In contrast to the existing approaches (where each algorithm is tailored to specific input metric and clustering objective), our algorithmic framework is (almost) entirely oblivious to any specific metric structures and the objective function, yet is able to effectively exploit those unknown structures. 

\subsection{Efficient Parameterized Approximation Schemes for \knc}

As an input to the (general) $k$-clustering problem, we are given $n$ data points $P$, candidate centers~$F$, a metric space $M=(P\cup F,\delta)$, a positive integer $k$, and an objective function $f\colon {\mathbb R}^P \rightarrow {\mathbb R}$. When a set of $k$ ``open'' centers $X\subseteq F$ is chosen, this solution induces a cost vector $\vcd(P,X)=(\delta(p,X))_{p\in P}$ where $\delta(p,X) = \min_{x \in X} \delta(p,x)$ represents the distance from point $p$ to the closest center in $X$. Our goal is to minimize $f(\vcd(P,X))$. We call this problem the $k$-clustering problem with cost function~$f$. We may think of the function $f$ as ``aggregating'' the costs incurred by the points.  
For example, we can formulate basic $k$-clustering  objectives via the functions $f(\vc{x}) = \sum_{p\in P}x(p)$ (\km), $f(\vc{x}) = \sum_{p\in P}x(p)^2$ (\kmns) and $f(\vc{x}) = \max_{p\in P}x(p)$ (\kc). 

Most natural and well-studied clustering objectives can be modeled using (a generalization of) the concept of \emph{norm} optimization introduced by Chakrabarty and Swamy~\cite{chakrabarty2019approximation}. More specifically, we are interested in the setting where the objective $f$ is a norm. A norm is a function $f\colon\bRna\rightarrow\bR_{\geq 0}$, $n\in\bN$ that satisfies (i) for all $\vc{x} \in \bRna$, $f(\vc{x}) =0$ if and only if $\vc{x}=\vc{0}$, (ii) $\forall \vc{x}, \vc{y} \in {\bRna}\colon f(\vc{x}+ \vc{y}) \leq f(\vc{x}) + f(\vc{y})$, and (iii) $\forall \vc{x} \in {\mathbb R}^n, \lambda \in {\mathbb R}\colon f(\lambda \vc{x}) = |\lambda| f(\vc{x})$. 
We say that $f$ is \textit{monotone} if $f(\vc{x}) \leq f(\vc{y})$ whenever $\vc{x} \leq \vc{y}$.  By \knc we refer to the $k$-clustering problem whose objective $f\colon \mathbb{R}^P \rightarrow \bR_{\geq 0}$ is a monotone norm.
While Chakrabarty and Swamy~\cite{chakrabarty2019approximation} further require that $f$ be symmetric\footnote{We say that $f$ is symmetric if $f(\vc{x}) = f(\vc{x'})$ whenever $\vc{x'}$ can be obtained by reordering coordinates of $\vc{x}$.}, our algorithmic framework applies to all monotone norm cost functions. 
This family includes the following well-known clustering problems (see Figure~\ref{fig:prob-diagram} for an overview): 
\begin{itemize}
    \item \textbf{From \kmns, \kc, and \km to \kzc}: All the basic clustering  problems can be captured by the $\ell_z$-norm when $z \in \{1,2, \infty\}$.  In fact, a   $(k,z)$-clustering problem~\cite{huang2020coresets,cohen2021new,cohen2022towards} (for constant positive integer $z$) uses the objective function $g(\vc{x}) = \sum_{p \in P} |x(p)|^z$. (This function itself is not a norm, but we can instead consider the $\ell_z$-norm $f(\vc{x})= g(\vc{x})^{1/z}$.)

    \item \textbf{\wkc (or \pkc)}: The weighted version of \kc \cite{li2003asymmetry,bajpai2021revisiting,plesnik1987heuristic} generalizes the \kc so that each data point $p \in P$ is associated with a positive weight (or priority) $w(p)$, and the objective is to minimize the (weighted) maximum distance to a center.\footnote{For convenience of presentation, the terminologies we use are somewhat different from the literature.} 
    This problem can be modelled by the ``weighted max'' norm $f(\vc{x}) = \max_{p \in P} w(p) x(p)$.
    One can analogously define the weighted versions of \km and \kmns (see, for example, \cite{DBLP:conf/icalp/Cohen-AddadG0LL19}).  We remark that the underlying weighted norms are not symmetric.

    \item \textbf{\lcentrum}: This problem (sometimes called \textsc{$k$-Facility $\ell$-Centrum}) aims to minimize the sum of the connection costs among the $\ell$ ``most expensive'' points (that is, those that are furthest away from the open centers). The problem
    generalizes both \kc ($\ell=1$) and \km ($\ell=|P|$) problem~\cite{tamir2001k}. (See the books~\cite{LocationTheory,LocationScience} for more details on \lcentrum and the more general \okm discussed below.)
    This problem can be modelled by the top-$\ell$ norm $f(\vc{x}) = \sum_{j=1}^{\ell} x^{\downarrow}(j)$ where $\vcx^{\downarrow}$ denotes the reordering of vector $\vc{x}$ so that the entries appear non-increasingly. The top-$\ell$ norm is symmetric.

    \item \textbf{\okm}: This problem further generalizes \lcentrum, allowing flexible penalties to be applied to data points that incur the highest connection costs. More formally, the objective is the \emph{ordered weighted norm} $f(\vc{x}) = \vc{v}^{\intercal} \vc{x}^{\downarrow}$ where $\vc{v}\in\bRn$ is a non-increasing cost vector, that is, $v(1) \geq v(2) \geq \ldots \geq v(n)$. \lcentrum corresponds to $\vc{v} = (1,\ldots, 1,0, \ldots 0)$ where the first $\ell$-entries of $\vc{v}$ are ones. This problem has already received attention for a few decades~\cite{byrka2018constant,chakrabarty2019approximation,braverman2019coresets}. We remark that the $f$ here is a monotone and symmetric norm.  

    \item \textbf{\sfm (or \rkm)}: In \sfm, along with the point set $P$, we are given $m$ different (not necessarily disjoint) subgroups such that $P=\bigcup_{i \in [m]} P_i$. Our goal is to find a set $X$ of centers that incurs fair costs to the groups by minimizing the maximum cost over all the groups. In other words, 
    \[\min_{\substack{X \subseteq F\\ |X|=k}} \max_{i \in [m]} \sum_{p\in P_i}\delta(p,X)\,.\]
    Due to  distinct applications in at least two domains, this variant of clustering has recently been studied extensively: (i) in algorithmic fairness~\cite{abbasi2021fair,goyal2023tight,pmlr-v134-makarychev21a,ghadiri2022constant} and (ii) in the robust optimization context, this problem is known as \rkm, which intends to capture the applications when we are uncertain about the actual data scenarios (corresponding to the groups $P_i$) that may come up~\cite{anthony2010plant,bhattacharya2014new,bansal2013generalizations}.   

    \item \textbf{\pqfc:}
    Our problem also models a clustering problem called \pqfc\footnote{Chlamt{\'a}{\v{c}} et al.~\cite{chlamtavc2022approximating} call the problem \textsc{$(p,q)$-Fair Clustering}. For the sake of consistency with the notation in the rest of the paper, we changed the naming slightly.} introduced by Chlamt{\'a}{\v{c}} et al.~\cite{chlamtavc2022approximating}, which generalizes \sfm. 

In particular, one can view the cost function $f$ of \sfm as a ``two-level'' aggregate cost: First, cost $\sum_{p\in P_i}\delta(p,X)$ incurred by group $P_i$, $i\in [m]$ can be viewed as weighted $\ell_1$-norm $\vc{w}_i^{\intercal} \vc{x}$  where $\vcw_i=\vc{1}_{P_i}\in\{0,1\}^P$ denotes the characteristic vector of $P_i$. Second, these group costs are further aggregated through $\ell_{\infty}$, that is, $f(\vc{x}) = \max(\vcw_1^{\intercal}\vcx, \vcw_2^{\intercal}\vcx, \ldots, \allowbreak\vcw_m^{\intercal}\vc{x})$. 

\pqfc allows arbitrary uses of $\ell_z$ and $\ell_q$ norms to aggregate the costs in two levels. The cost function is defined as $f(\vc{x}) = g(\vc{h}(\vc{x}))$ where $g$ is any $\ell_q$-norm function and  $\vc{h}(\vc{x}) = (h_1(\vc{x}), h_2(\vc{x}),\ldots, h_m(\vc{x}))$ where $h_i(\vc{x})$ is a weighted $\ell_z$-norm, that is, $h_i(\vc{x})= \left(\sum_{p\in P}w_i(p)x(p)^z\right)^{1/z}$ for arbitrary weight vectors $\vcw_i \in \bRp$, $i\in[m]$. 
    It is easy to check that $f(\vc{x}) = g(\vc{h}(\vc{x}))$ is a monotone norm whenever $g$ and $\{h_i\}$ are. 
    
    \item \textbf{Beyond the Known Problems}: 
Our (asymmetric) norm formulation allows us to model more complex clustering objectives that might be useful in some application settings and, to our knowledge, have not yet been considered in the algorithms community.  
    One such objective is the {\sc Priority Ordered $k$-Median}: We have the cost function $f(\vc{x}) = \vc{v}^{\intercal}{\vcx_{\vcw}}^{\downarrow}$ where the weight vector $\vc{v}\in\bRn$, and priority vector $\vc{w} \in \bRp$ are given as input, and where $\vcx_{\vcw}=(w(p)x(p))_{p\in P}$. 
    This objective generalizes both \pkc and \okm.
    Another natural objective is the (multi-level) {\sc Cascaded Norm Clustering}, which generalizes \pqfc to allow multiple levels of cost aggregation.     
    The cost function $f$ for this problem is described by a directed acyclic graph (DAG) $D$ with one sink node and $|P|$ source nodes (each source corresponds to a point in $P$). Each non-source node $v$ is associated with a norm $\ell_q$ for some $q$, and each edge $(u,v)$ has weight $w_{u,v}$. 
    Given such a DAG $D$,  the value of $f(\vc{x})$ can be evaluated by computing the evaluations at nodes in $V(D)$ in (topological) order from sources to sink: (i) The evaluation at  source $p \in P$ is $\eta(p) = x(p)$, (ii) For any non-source node $v \in V(D)$ labelled with the norm $\ell_q$, we evaluate $\eta(v) = \left(\sum_{u \in N^-(v)} w_{u,v} \eta(u)^q\right)^{1/q}$, and (iii) the value of $f(\vc{x})$ is the evaluation of the sink. See Figure~\ref{fig:cascaded} for illustration. \pqfc  is a special case when $D$ has 3 layers with the middle layer using the same norm. Of course, also other basic monotone norms such as top-$\ell$ or ordered weighted norms could be composed to more complex norms analogously.

 \begin{figure}[t]
  \begin{center}
    \includegraphics[width=0.45\textwidth]{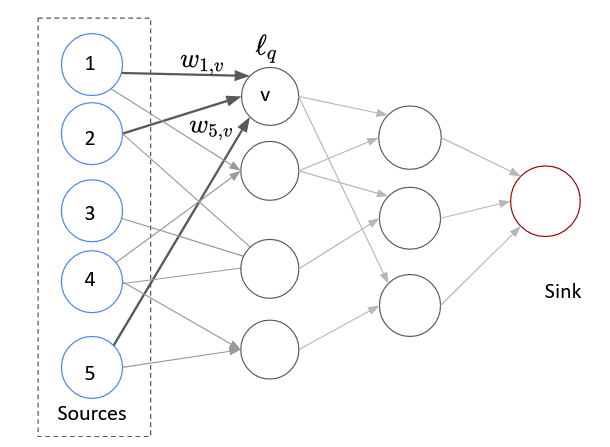}
  \end{center}
  \caption{The DAG here describes evaluation of function $f$.  Node $v$ is labeled with the $\ell_q$ norm, so the evaluation at node $v$ is $\eta(v) = (w_{1,v} x_1^q + w_{2,v} x_2^q + w_{5,v} x_5^q)^{1/q}$.\label{fig:cascaded}}
\end{figure}

\end{itemize}

We remark that \emph{asymmetric} norms  can potentially make the problem substantially harder. For example, a poly-time $\OO(1)$-approximation algorithm exists for symmetric norms~\cite{chakrabarty2019approximation} but the asymmetric norm makes it $\Omega(\log n/\log \log n)$-hard to approximate even for the special case of \rkm on  the line metrics~\cite{bhattacharya2014new}.

Our main  results are encapsulated in the following theorem.  

\begin{theorem}\label{thm:main-informal}
Let $f$ be an efficiently computable monotone norm cost function. Then the $k$-clustering problem with cost function $f$ admits an EPAS for the following input metrics: (i) metrics of bounded doubling dimension, (ii) continuous Euclidean spaces of any dimension, (iii) bounded treewidth metrics, and (iv) planar metrics. 
\end{theorem}

By \emph{continuous} Euclidean space, we refer to the setting where any point of the space can be chosen as a center. This is in contrast to a \emph{discrete} Euclidean space, where we restrict the centers to be selected from a specific finite subset of the points.Observe that for a fixed $d$, discrete Euclidean problems in $\bRd$ have bounded doubling dimension, hence covered by our framework. Furthermore, it is not a shortcoming of our result that it does not cover discrete Euclidean spaces of high dimensions: in this setting, \kc  
is W[1]-hard to approximate within a factor of $\sqrt{3/2} -o(1)$  (the proof of this will appear elsewhere).

Our result in particular implies the following.

\begin{corollary}
In all aforementioned metric spaces,
\begin{enumerate}
    \item There exists a $2^{h(1/\epsilon)\cdot k\cdot\polylog (k)}\cdot \poly(n)$ time EPAS  for \okm on $n$ points.
\item There exists a $2^{h(1/\epsilon)\cdot k\cdot\polylog(k)}\cdot \poly(n,m)$ time EPAS  for \pqfc on $n$ points and $m$ groups.
\end{enumerate}
\end{corollary}

Prior to our results, the existences of EPASes for all these problems were open (except for \kmns , \kc , and \km). Beyond these known problems, we also obtain EPASes for the new, generalized problems introduced above and depicted in Figure~\ref{fig:prob-diagram}. Rather surprisingly, in contrast to the poly-time approximation regime, the complexities of symmetric and asymmetric norm clustering problems ``collapse'' in the parameterized approximation regime.

\begin{figure}[t]
    \begin{center}
    \includegraphics[width=0.6\textwidth]{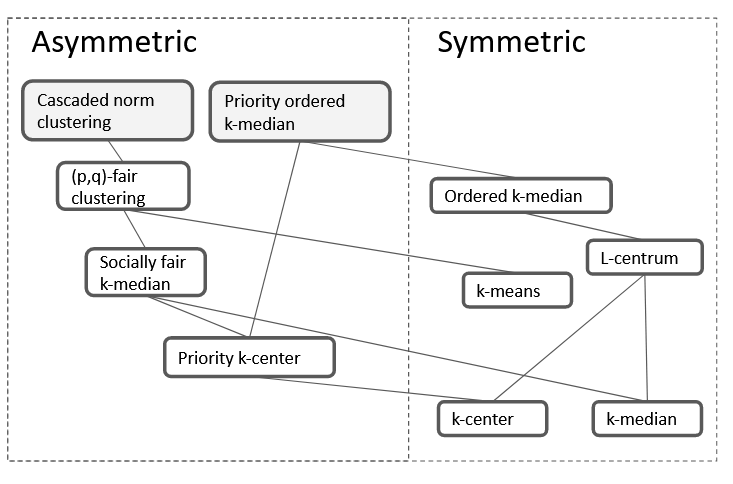}
    \end{center} 

    
    \caption{Selected clustering objectives that can be formulated as monotone norm minimization. The line illustrates generalization  (bottom is a special case of top).  \label{fig:prob-diagram} } 
\end{figure}

\subsection{Our Conceptual and Technical Contributions}

Our main contributions have two parts: (i) a new concept of metric dimension and (ii) our main technical result showing EPASes for all the aforementioned clustering problems. 

\subsubsection*{Unifying Metric Spaces via Scatter Dimension} 

Our key conceptual contribution is a new notion of bounded metric space dimension that relaxes  the standard requirement of bounded doubling dimension so that the metric spaces mentioned in Theorem~\ref{thm:main-informal} all ``live'' in a finite dimension. 
We first explain why existing notions of dimensions are not suitable for such purpose. 

There are multiple dimensionality notions that appear in the literature of  metric spaces. Most familiar in the algorithmic community is perhaps  the  \textit{doubling dimension} (a.k.a. Assouad dimension). Roughly, the doubling dimension of metric $(M,\delta)$ is $\OO(d)$ iff at most $(\nicefrac{1}{\epsilon})^{\OO(d)}$ balls of radius $\nicefrac{\epsilon}{2}$ can be packed into a unit ball (this is called an $\epsilon$-packing). 
Such property can often be computationally leveraged, leading to efficient algorithms for many geometric optimization problems (often with the running time depending exponentially on the dimension). 
However, the doubling dimension (as well as any other popular notions of dimensions~\cite{clarkson2006nearest}) would not be suitable for us due to the following reasons: (i) The doubling dimension can be as large as $\Omega(n)$ in high-dimensional Euclidean space, and (ii) they do not very well ``exploit'' structured graph metrics, i.e., even stars have unbounded dimension.\footnote{In an $n$-node star rooted at $r$, a unit ball $\ball(r,1)$ includes the whole graph. There exists an $\epsilon$-packing of size $(n-1)$ by choosing the non-root nodes.}    In sum, any algorithms that exploit existing notions of dimensions are unlikely to lead to our desired results.

We introduce the notion of \emph{\scatterdim}. Given metric $M = (P,F,\delta)$,  the sequence $(x_1,p_1),\ldots, (x_\ell,p_{\ell}) \in F\times P$ is said to be an $\epsilon$-scattering if, whenever $(x,p)$ appears before $(x',p')$ in the sequence, then $\delta(x,p)$ and $\delta(x',p')$ are larger than $1+\epsilon$ each, while $\delta(x',p) \leq 1$. 
The \scatterdim of $M$ is then defined as the length of the longest scatter, minus one. 

There are two natural interpretations. The first interpretation is as a   game between two players: The \textit{center player} who tries to claim she can cover all the points with a unit ball and the \textit{point player} who present a counterexample. In the first round, the center player picks a center $x_1 \in F$ and the point player \textit{refutes} the claim by presenting a point $p_1 \in P$ which is at least a factor $1+\epsilon$ away from the (closed) unit ball around $x_1$, that is, $p_1 \not\in \ball(x_1,1+\epsilon)$.
The game continues this way: In the $i$-th round, the center player presents $x_i$ such that $\{p_1,\ldots, p_{i-1}\} \subseteq \ball(x_i,1)$, and the point player gives $p_i \not \in \ball(x_i,1+\epsilon)$.  Both players are interested in prolonging the game as much as possible. The \scatterdim is the length of the longest possible game.  
In the second interpretation, one can view such sequence as a pair of $\epsilon$-packings  that are required to be sufficiently distanced: It is easy to verify (simply using triangle inequalities) that $P^* = \{p_1,p_2,\ldots, p_{\ell-1}\}$ and $F^*= \{x_2,\ldots, x_{\ell}\}$ are $\epsilon$-packings of the unit (closed) balls around $x_{\ell}$ and $p_1$, respectively. This view immediately implies that \scatterdim is bounded in a bounded doubling metric. 

\begin{theorem}
\label{thm:doubling-metric}
For $\epsilon \in (0,1)$, any metric of doubling dimension $d$ has \scatterdim $(1/\epsilon)^{\OO(d)}$.  
\end{theorem}

We proceed to study the \scatterdim of bounded treewidth graphs. 

\begin{theorem}
\label{thm:treewidth}
For $\epsilon \in (0,1)$, the \scatterdim is  $\exp\left((1/\epsilon)^{\OO(\tw)}\right)$ for treewidth-$\tw$ graphs. 
\end{theorem}

This proof is based on a (delicate) combinatorial argument that, given graph $G$, parameter $t$ and an \scattering sequence of length at least doubly exponential in $t$, produces a ``certificate'' to the fact that the treewidth of $G$ is greater than $t$. The proof can be found in Section~\ref{sec:graph-metrics}. 

Next, we present a tool that allows ``bootstrapping'' of graph classes having bounded \scatterdim. This is done via a simple connection between \scatterdim and low-treewidth embedding (an active area of metric space embedding)~\cite{filtser2022low,fox2019embedding,cohen2020light}. This connection would allow us to reduce the question of bounding \scatterdim in a certain graph class to that in  bounded treewidth graphs (thereby invoking our Theorem~\ref{thm:treewidth}.) 

\begin{theorem}[informal, formal statement in Section~\ref{sec:planar}]
\label{thm:embedding and scatter dim}
The \scatterdim is bounded for any graph class ${\mathcal G}$ that admits an $\eta$-additive distortion embedding (error $\pm \eta \Delta$ where $\Delta$ is the diameter of the graph) into a graph whose treewidth only depends on $\eta$.    
\end{theorem}
 
Such a connection, combined with the  embedding result of~\cite{fox2019embedding},  implies the following.

\begin{theorem}
\label{thm:planar}
For $\epsilon \in (0,1)$, the \scatterdim is  $\exp\left(\exp({\sf poly}(\nicefrac{1}{\epsilon}))\right)$ for planar graphs. 
\end{theorem}

Moreover, further progresses in the area of low-treewidth embedding would  lead to even wider classes of graphs that have bounded \scatterdim, e.g., it seems plausible that minor-free graphs admit such an embedding~\cite{filtser2022low}.  

Unfortunately, the bounded dimensionality does not hold in the high-dimensional (continuous) Euclidean metric.\footnote{To see this, consider the  sequence $(x_1,p_1)\dots, \allowbreak (x_{d-1},p_{d-1})$ where, for each $i \in [d-1]$, the point $x_i\in\bR^d$ has $i$-th coordinate  $1/\sqrt{2}$ and all other coordinates are zero. Define points $p_i=-x_i$ for all $i \in [d-1]$. 
It is easy to check that this sequence is a  \scattering[$(\sqrt{2}-1)$].
This example  implies that the \scatterdim of continuous Euclidean metrics ${\mathbb R}^d$ can be at least $d-1$ (unbounded in $\epsilon$). In fact, the \scatterdim is as high as $(\nicefrac{1}{\epsilon})^{\Omega(d)}$.} To handle the high-dimensional continuous Euclidean setting, we present a stronger version of \scatterdim, that we call  \emph{\algscatterdim}. The setting of the game is the same except that the center player would \emph{optimize} to end the game early, while the point player would  be interested in prolonging the game indefinitely. This means, they play against each other. 
A \textit{centering strategy} is a function $\sigma\colon 2^{P} \rightarrow F$ that specifies how the center $x_i = \sigma(\{p_1,\ldots, p_{i-1}\})$ would be chosen by the center player, given the points $p_1,\ldots, p_{i-1}$ played in the preceding rounds. 
The \scatterdim[$(\sigma,\epsilon)$] is the maximum number of rounds when the center player always plays strategy $\sigma$, and the \algscatterdim is the minimum \scatterdim[$(\sigma,\epsilon)$] over all strategies $\sigma$. We remark that our actual definition is more involved, as it considers a weighted version of the game.

\begin{restatable}[Bounding Algorithmic Scatter Dimension]{theorem}{conteuclscatter}\label{thm:scatter-cont-eucl}
The continuous Euclidean space $(P,F,\delta)$, that is, $P\subsetneq\bR^d$ finite, and $F=\bR^d$, has \algscatterdim $\OO(\nicefrac{1}{\epsilon^4}\log\nicefrac{1}{\epsilon})$.
\end{restatable}

\subsubsection*{EPAS for General Norm Clustering: Bypassing Coresets}
Now we are ready to explain our main technical result that would allow us to obtain EPAS for all metrics having bounded \scatterdim.

A generic tool whose existence immediately implies an EPAS is an $\epsilon$-coreset 
---a ``compression'' of an input instance $(P,F,\delta)$ into a much smaller instance so that the cost of any solution is preserved within a factor of $(1\pm \epsilon)$. The existence of an $\epsilon$-coreset of size depending only on $\epsilon$ and $k$ would immediately imply an EPAS (but not vice versa): First, use the $\epsilon$-coreset to compress the instance $(P,F,\delta)$ to $(P',F', \delta')$ where $|P'| \leq \gamma(\epsilon,k)$. Then enumerate all possible partitionings of $P'$ into $k$ sets $P'_1,\ldots, P'_k$ (there are at most $k^{\gamma(\epsilon,k)}$ such partitions). For each set $i \in [k]$, compute the optimal center for $P'_i$. We choose the partition that gives the lowest total cost.

This generic method, unfortunately, faces a serious  information-theoretic limitation, that is, even for \kc,  $\epsilon$-coresets of desirable sizes do not exist in high-dimensional Euclidean spaces~\cite{feldman2011unified}. Such lower bounds imply that one cannot hope to prove our (unified) results via the coreset route: While coresets are known for \kzc for constant $z$~\cite{cohen2021new}---allowing to handle \kmns and \km in a uniform fashion---it is impossible to extend this approach to \kc.
For more complex clustering objectives, such EPASes were in fact not known even for low dimension. For example, the coreset of Braverman et al.~\cite{braverman2019coresets} for \okm in $\bRd$ has size $\OO_{\epsilon,d}(k^2 \log^2 n)$ and therefore does not give an EPAS even in low dimension.

Badoiu, Har-Peled, and Indyk~\cite{badoiu-etal:approximate-clustering-coresets} presented an EPAS for \kc in high-dimensional Euclidean spaces (bypassing coresets in the above sense).  
Therefore, an obvious open question is whether their techniques can be extended to give an EPAS for any other clustering objective. Unfortunately, this is not even known for simple objectives such as \pkc. In fact, even 
the known  EPASes for \kmns~\cite{kumar2010linear} and \kc~\cite{badoiu-etal:approximate-clustering-coresets} are conceptually very different; to our knowledge, no approximation schemes   handle \kmns and \kc in a modular way. 

Our main technical result is presented in the following theorem. 
We remark that our techniques do not rely on any coreset constructions (thus bypassing the coreset lower bounds for \kc).

\begin{restatable}{theorem}{thmmain}\label{thm:main}\label{thm:main-semi-formal}

Let $\cM$ be a class of metric spaces that is closed under scaling distances by a positive constant. There is a randomized algorithm that  computes for any \knc instance $\cI=(M,f,k)$ with metric $M=(P,F,\delta)\in \cM$,  and any $\epsilon\in (0,1)$, with high probability a $(1+\epsilon)$-approximate solution if the following two conditions are met.

 \begin{enumerate}[(i)]
      \item There is an efficient algorithm evaluating for any distance vector $\vc{x}\in\bRp$ the objective $f(\vc{x})$ in time $T(f)$. 

     \item There exists a function $\lambda\colon \bR_{+} \rightarrow \bR_{+}$, such that for all $\epsilon >0$, the \algscatterdim of $\cM$ is at most  $\lambda(\epsilon)$.
 \end{enumerate}
 The running time of the algorithm is $\exp{\left(\widetilde{\OO}\left(\frac{k\lambda(\nicefrac{\epsilon}{10})}{\epsilon}\right)\right)}\cdot\poly(|M|)\cdot T(f)$.

\end{restatable}

Note that the complexity of computing $f$ appears only as a linear factor in the running time. For instance, for \sfm, the number $m$ of groups affect only the computational cost of $f$, and therefore the running time is polynomial in $m$.

Our algorithm is clean, simple, and entirely oblivious to both the objective and the structure of the input metric.

The dependency on $k$ in the exponent of our running time is  singly exponential ($\exp({\widetilde{\OO}}_{\epsilon}(k))$). In terms of $k$, we therefore match the running time of the fastest known EPAS for the highly restrictive special case of high-dimensional \kmns~\cite{kumar2010linear}. Moreover,  the dependency on $\epsilon$ in the exponent could be improved by proving better bounds on the \scatterdim of a metric space of interest, e.g.,  $\lambda(\epsilon) = {\sf poly}(1/\epsilon)$ implies the EPAS running time $\exp(\widetilde\OO(k) \cdot {\sf poly}(1/\epsilon))$.  

\section{Overview of Techniques} 
\label{sec:overview}

In this section, we give an informal overview of the technical ideas appearing in the paper. The main result will be built step by step: we believe that it is already interesting to understand our main result specialized to \wkc and \wkm. 
Our starting point is the EPAS of Bad\u{o}iu et al.~\cite{badoiu-etal:approximate-clustering-coresets} for unweighted \kc that works on high-dimensional Euclidean spaces. We redesign and change this algorithm in order to be able to present it with a clean division into two parts: a simple branching algorithm and a bound on the abstract concept of (algorithmic) \scatterdim. This way, we obtain a sharp separation between the branching algorithm, which is specific to the objective and the bound on \scatterdim, which is specific to the metric. This can be contrasted with techniques based on coresets, which are inherently specific both to a single objective and to a single metric. The main message of the paper is that, with the right combination of additional ideas, this framework can be significantly generalized both in terms of objectives and metric spaces.

This section presents the main algorithmic ideas in three steps.

\begin{enumerate}
\item The algorithm for unweighted \kc can be generalized to \wkc  in a not completely obvious way.
\item Building on the algorithm for \wkc, we can solve \wkm with a preprocessing and a random selection step.
\item The \wkm algorithm can be generalized to arbitrary monotone norms by considering infinitely many \wkm instances defined by the subgradients.
    \end{enumerate}

While some of the challenges on the way may appear to have other approaches promising at first glance, we want to emphasize that it is nontrivial to find the combination of ideas that can be integrated together to obtain our main result. In particular, for \wkm the initial upper bounds have to be defined carefully in a way that allows, at the same time, an efficient random selection step and generalization to arbitrary monotone norms.

\newcommand{\iletter}{\kappa}

\paragraph{\wkc with Bounded Number of Different Weights.}
Our starting point is a simple branching algorithm that is inspired by the EPAS  of Bad\u{o}iu et al.~\cite{badoiu-etal:approximate-clustering-coresets} for unweighted \kc. Instead of branching, it will be more convenient for us to present it as a randomized algorithm. Furthermore, we consider the more general setting of \wkc: the objective is to find a set $O$ of $k$ centers that minimizes $\max_{p}w(p)\delta(p,O)$. Let us first present the algorithm with the simplifying assumption that $w$ is a weight function on the points whose range contains only at most $\tau$ different values. The unweighted problem corresponds to $w(p)=1$ for every $p\in P$ and hence $\tau=1$. It will be convenient to assume that we (approximately) know the value of $\opt$.

We start with $k$ arbitrarily chosen candidates $X=\{x_1$, $\dots$, $x_k\}$ for the $k$ centers. We additionally introduce $k$ sets of \emph{requests} $Q_1$, $\dots$, $Q_k$, where each request is of the form $(p,r)$ with a point $p\in P$ and radius $r>0$. For every $\iletter\in[k]$, we impose the \emph{cluster constraint} requiring that, for every $(p,r)\in Q_\iletter$, center $x_\iletter$ should be at distance at most $r$ from~$p$. Initially, we set $Q_\iletter=\emptyset$ for every $\iletter$, which means that these conditions are trivially satisfied. If we have $\max_{p}w(p)\delta(p,X)> (1+\epsilon)\opt$,  then we can stop, as we have a $(1+\epsilon)$-approximate solution at our hands. Otherwise, we have a point $p$ with $\delta(p,X)\ge (1+\epsilon)\opt/w(p)$, while it is at distance at most $\opt/w(p)$ from some center of a hypothetical optimum solution $O$. Thus the algorithm selects a $\iletter\in [k]$ uniformly at random, hoping it to be the index of the center that is at distance at most $\opt/w(p)$ from $p$ in the optimum solution $O$. Then we introduce the request $(p,\opt/w(p))$ into the set $Q_\iletter$ and select $x_\iletter$ to be a center that satisfies the cluster constraint defined by all the requests in the updated $Q_\iletter$. 
Observe that if every random choice was compatible with the hypothetical optimum solution $O$, then the algorithm is always able to find such a center, as the requests in $Q_\iletter$ are always satisfied by the $\iletter$-th center of the optimum solution $O$. 

We claim that if the \scatterdim of the metric is bounded, then this algorithm stops after a bounded number of steps, either by finding an approximate solution or by failing to find a center satisfying the cluster constraints of some $Q_\iletter$.
Let $x^{(1)}_\iletter$, $\dots$, $x^{(\ell)}_\iletter$ be the different candidates for the $\iletter$-th center throughout this branch. Let $(p^{(1)}_\iletter,r^{(1)}_\iletter)$, $\dots$, $(p^{(\ell)}_\iletter,r^{(\ell)}_\iletter)$ be the requests introduced to $Q_\iletter$: that is, for $1\le j\le \ell$, the center $x^{(j)}_\iletter$ was chosen to be at distance at most $r^{(i)}_\iletter$ from every $p^{(i)}_\iletter$ for $1 \le i < j$,  but later was found to be at distance at least $(1+\epsilon)r^{(j)}_{\iletter}$ from $p^{(j)}_{\iletter}$.
As there are at most $\tau$ different weights in the input, at least $\ell'=\ell/\tau$ of these requests have the same radius. That is, there is a subsequence $(x^{(s_1)}_\iletter,p^{(s_1)}_\iletter,r^{(s_1)}_\iletter)$, $\dots$, $(x^{(s_{\ell'})}_\iletter,p^{(s_{\ell'})}_\iletter,r^{(s_{\ell'})}_\iletter)$ where every $r^{(s_{j)}}_\iletter$ for $j\in[\ell']$ is the same value $r\ge 0$.
 This means that we have a subsequence $(\bar x_1,\bar p_1)$, $\dots$, 
$(\bar x_{\ell'},\bar p_{\ell'})$ with the property that $\delta(\bar x_i,\bar p_i)>(1+\epsilon)r$, but $\delta(\bar x_i,\bar p_j)\le r$ for every $i<j$. By scaling down every distance by a factor of $r$, this is precisely an \scattering of length $\ell'$. If we consider a class of metrics closed under scaling where the \scatterdim is $\lambda(\epsilon)$, then this sequence cannot have length longer than $\lambda(\epsilon)$, implying that $\ell\le \tau \cdot \lambda(\epsilon)$. We can conclude that the algorithm can introduce at most $\tau \cdot \lambda(\epsilon)$ requests into each $Q_\iletter$, hence the algorithm cannot perform more than $k\cdot \tau\cdot \lambda(\epsilon)$ iterations. 

If every step of the algorithm randomly chooses an index $\iletter\in [k]$ that is consistent with the optimum solution $O$, then the only way it can stop is by finding an approximate solution. Therefore, the algorithm is successful with probability at least $q=k^{-k \cdot \tau\cdot \lambda(\epsilon)}$. The success probability can be boosted to be a constant arbitrarily close to 1 by the standard technique of repeating the algorithm $\OO(1/q)$ times, leading to a running time of $k^{k \cdot \tau\cdot \lambda(\epsilon)}\cdot \poly(n)$.

\paragraph{\wkc with Arbitrary Weights.}
  We show now how the algorithm can be extended to work in the weighted setting with arbitrary weights. Let us observe first that if there is no bound on the number $\tau$ of different weights, then we cannot bound the number of requests to a given $Q_\iletter$, even in very simple metric spaces such as $\mathbb{R}^1$. Suppose for example that the requests arriving to $Q_\iletter$ are $(p^{(i)},(1+2\epsilon)^{-i})$ for $i=1,2,\ldots$, where every $p^{(i)}$ is at the origin (or maybe within a very small radius of the origin). Then a center $x^{(i)}$ at $(1+2\epsilon)^{1-i}$ satisfies the first $i-1$ requests, but violates the constraint of the $i$-th by more than a $(1+\epsilon)$-factor.
  This sequence can be arbitrarily long, and the existence of such a sequence shows that we cannot bound the number of requests arriving to $Q_\iletter$ if we don't have a bound on the number of different weights. Nevertheless, we show that the number of requests can be bounded if we start the algorithm by carefully seeding the initial requests. Let us remark that we know other simple modifications that achieve such a bound, but the technique described below turns out to be the one that can be extended further for $\wkm$ and general norms.

The main idea is to bootstrap our algorithm with a constant-factor approximation. A simple greedy 3-approximation can be obtained following the ideas of Plesn\'\i k \cite{plesnik1987heuristic}. Let us consider all the balls $\ball(p,\opt/w(p))$ for every $p\in P$. Let us consider these balls in a nondecreasing order of radius, and mark each ball that does not intersect any of the balls marked earlier; let $\ball(p_\iletter,\opt/w(p_\iletter))$, $1\le \iletter\le k'$ be the marked balls. We should have $k'\le k$: otherwise, we have more than $k$ pairwise disjoint balls and each of them has to contain a center of the solution, contradicting the assumption that value $\opt$ can be achieved with $k$ centers. 
For $1\le \iletter \le k'$, let $x_i$ be any center in $\ball(p_\iletter,\opt/w(p_\iletter))$ and let $Q_\iletter=\{(p_\iletter,\opt/w(p_\iletter)\}$. For $k'<\iletter \le k$, we choose $x_i$ arbitrarily and let $Q_\iletter=\emptyset$. Let us observe that with this definition of the $Q_\iletter$'s, we have $\delta(p,X)\le 3\opt/w(p)$ during \emph{every iteration} of our algorithm. Indeed, if the ball of $p$ was marked, then $X$ always contains a center in $\ball(p_\iletter,\opt/w(p_\iletter))$; if the ball of $p$
was unmarked, then it intersects a marked ball with not larger radius that contains a center of $X$.

The main claim is that the ratio between the radii of two requests appearing in $Q_\iletter$ can be bounded by $\OO(1/\epsilon)$.
Suppose that $(p,r)$ and $(p',r')$ are two requests in $Q_{\iletter}$ (introduced in any order) and we have $r'<\epsilon r/4$. A center of the optimum solution satisfies both request, hence we have $\delta(p,p')\le r+r'$.  As shown above, at every step of the algorithm there is a center in $X$ at distance at most $3r'$ from $p'$; let $y$ be such a center at the step when request $(p,r)$ was introduced. 
Then we have 
  \[
\delta(p,y)\le \delta(p,p')+ \delta(p',y) \le r+r'+3r' \le (1+\epsilon)r, 
\]
contradicting the need for the first request.
 
We can use the standard assumption that every weight is of the form $(1+\epsilon)^i$ for some integer~$i$: by rounding down every weight to the largest number of this form, we change the objective only by a factor of $1+\epsilon$. If every weight is of the form $(1+\epsilon)^i$, then the $\OO(1/\epsilon)$ bound proved above implies that the requests introduced into $Q_{\iletter}$ for some fixed $\iletter\in [k]$ have $\OO(1/\epsilon \cdot \log 1/\epsilon)$ different radii. Therefore, we can bound the total number of requests (and hence the number of iterations) by $\OO(\lambda(\epsilon)\cdot k/\epsilon \cdot \log 1/\epsilon)$. 
This leads to a $k^{\OO(\lambda(\epsilon)\cdot k/\epsilon \cdot \log 1/\epsilon)}\cdot \poly(n)$ time randomized algorithm with constant success probability.

\paragraph{From \wkc to \wkm.} Towards our goal of understanding general norms, let us consider now the \wkm problem, where 
the objective is to find a set $O$ of $k$ centers that minimize $\sum_{p}w(p)\delta(p,O)$. We will try to solve this problem by interpreting it as a \wkc problem on a weighted point set that we dynamically discover during the course of the algorithm.

We would like to turn the \emph{linear constraint} $\sum_{p}w(p)\delta(p,X)\le \opt$ of \wkm into a \emph{distance constraint:} some point $p$ should be at distance at most $r$ to the solution. 
Let $X$ be the current solution and suppose that $\sum_{p}w(p)\delta(p,X)> (1+\epsilon)\opt$. 
The intuition is that $\sum_pw(p)\delta(p,X)> (1+\epsilon)\sum_pw(p)\delta(p,O)$ for an optimum solution $O$ implies that a nontrivial fraction of the points should satisfy $\delta(p,X)> (1+\epsilon/3)\delta(p,O)$, that is, their distances to the solution has to be improved by more than a factor of $1+\epsilon/3$. More precisely, an easy averaging argument shows if we select  a point $p$ with probability proportional to $w(p)\delta(p,X)$, then $p$ satisfies $\delta(p,X)> (1+\epsilon/3)\delta(p,O)$ with probability $\Omega(\epsilon)$. We call such a point $p$ an \emph{$\epsilon/3$-witness,} certifying that the current solution has to be improved.

Assuming that the sampled point $p$ is indeed a $\epsilon/3$-witness, we proceed as in the case of \wkc. We randomly choose an index $\iletter$ and introduce the request $(p,\delta(p,X)/(1+\epsilon/3))$ into $Q_\iletter$, 
to update the cluster constraint by requiring that $x_\iletter$ should be closer to $p$ than in the current solution. If there is a center satisfying all the requests in $Q_\iletter$, then we update $x_{\iletter}$. These steps are repeated until we arrive to a solution  $P$ with $\sum_{p}w(p)\delta(p,X)\le (1+\epsilon)\opt$.

In each step, with probability $\Omega(\epsilon/k)$, the algorithm chooses an $\epsilon/3$-witness $p$ and a center $\iletter$ that is consistent with some hypothetical optimum solution $O$. However, it is not clear how to bound the running time of the algorithm. It can happen that the requests arriving to $Q_\iletter$ have smaller and smaller radii. As we have seen for \wkc, in such a scenario we cannot bound the number of steps even in $\mathbb{R}^1$
It is crucial to have some control on the sequence of radii that appear in the requests. Therefore, next we show how to ensure that the radii in the requests to center $\iletter$ stay within a bounded range.   

\paragraph{Initial Upper Bounds.}
For each point $p$, we compute a weak upper bound $u(p)\ge \delta(p,O)$ on the distance to the optimum solution. Then instead of starting with an arbitrary set of $k$ centers, we bootstrap the algorithm by a solution approximately satisfying all these upper bounds. We argue that this can be done in such a way that ensures that the radii appearing in the requests to each center $\iletter$ stay within a bounded range.

If a point $p^*$ has weight $w(p^*)$, then $u(p^*)=\opt/w(p^*)$ is an obvious upper bound on the distance of $p^*$ to $O$: otherwise, we would have $\sum_{p}w(p)\delta(p,O)\ge w(p^*)\delta(p^*,O)>\opt$. This bound was sufficient for the \wkc problem, but the nature of \wkm allows us to get much stronger upper bounds in many cases. For example, if there are $c$ points of the same weight $w$ roughly at the same position, then each of them should be at distance at most $\opt/(wc)$ from~$O$. Indeed, otherwise the total contribution of these $c$ points to the sum would be greater than~$\opt$. More generally, if there is a radius $r$ such that total weight of the points at distance at most $r$ from $p$ is at least $\opt/r$, then we claim that $p$ is at most distance $2r$ from $O$. Indeed, otherwise all these points would be at distance more than $r$ from $O$, making their total contribution greater than $\opt$. Therefore, we can define 
$u(p)=2r$, where $r$ is the smallest radius with the property that the total weight of the points at distance at most $r$ from $p$ is at least $\opt/r$. Note that $u(p)$ can be determined in polynomial time from the weights of the points and their distance matrix.

Similarly to our \wkc algorithm, we start with a 3-approximation of the constraints given by the upper bounds $u(p)$ for $p\in P$. Let us go through the points in a nondecreasing order of $u(p)$ and let us greedily choose a maximal independent set of the balls $\ball(p,u(p))$. We should find at most $k$ such balls. Let us choose a center in each ball; it is easy to see that every point $p$ has a selected center at distance at most $3u(p)$ from it.  If center $x_\iletter$ was selected to be a center in $\ball(p,u(p))$, then we initialize $Q_\iletter$ with the request $(p,u(p))$. This ensures that during every step of the algorithm, it remains true that every point $p$ is at distance at most $3u(p)$ from the current solution.

We run the algorithm for \wkm with this initial solution. Before analyzing the algorithm, let us make a nontrivial change in the random selection. We have seen that with probability $\Omega(\epsilon/ k)$, we select a random point $p$ and $\iletter\in [k]$ such that $\delta(p,X)\ge (1+\epsilon/3)\delta(p,O)$ for some optimal solution $O$. A key claim of the proof is that with probability $\Omega(\epsilon/k)$, it is also true that $u(p)\le 2k\delta(p,X)/\epsilon$ (see Lemma~\ref{lem:sample-witness}). Intuitively, the total contribution of the $\epsilon/3$-witnesses that are too close to some center $x_\iletter\in X$
cannot be very large, because then all of these witnesses would be in a small ball, implying that the upper bound $u(p)$ should be smaller.
 Note that this is the point in the proof where we crucially utilize the exact definition of $u(p)$. With this claim at hand, we can modify the algorithm such that we are randomly choosing a point $p$ satisfying $u(p)\le 2k\delta(p,X)/\epsilon$, with probability proportional to $w(p)\delta(p,X)$. It remains true that $p$ is an $\epsilon/3$-witness with probability $\Omega(\epsilon/k)$.

Let us analyze now the algorithm and bound the number of times a center $x_{\iletter}$ is updated. We want to argue that the radius in the requests remains in a bounded range. Suppose that we update cluster $\iletter$ with requests $(p,r)$ and $(p',r')$ (in either order) such that $r'\ll \epsilon^2 r/k$. If the algorithm does not fail, then there is a center $x_\iletter$ satisfying both requests. By the triangle inequality, this means that the $\delta(p,p')\le r+r'<r+\epsilon r/6$. Furthermore, by the constraint $u(p')\le 2k\delta(p',X)/\epsilon= 2k (1+\epsilon/3)r'/\epsilon$
on our selection of the random point $p'$, we have that $u(p')$ is much smaller than $\epsilon r/18$. 
 At every step of the algorithm, the upper bound $u(p')$ is 3-approximately satisfied by the current solution $X$. Thus there should be a center in $X$ much closer than $\epsilon r/6$ to $p'$. Together with $\delta(p,p')< r+\epsilon r/6$, it follows that there is always a center in $X$ at distance at most $(1+\epsilon/3)r$ from $p$, contradicting the need for the request $(p,r)$.
 
  Thus the combination of the two facts that (1) the upper bounds are always satisfied approximately and that (2) the radius in the request is not much smaller than the upper bound implies that the radius in the requests stays within a bounded range. Then we can argue as in the case of the \wkc problem. If every weight is rounded to a power of $(1+\epsilon)$, then each cluster is given requests with only a bounded number of different radii. If many requests arrive, then there is a long subsequence of the requests with the same radius. This means that the bound on the \scatterdim can be used to bound the length of this subsequence, and hence the total number of requests to all clusters.

\paragraph{From \wkm to General Norms Using Subgradients.}
Next we show how to solve the clustering problem for an arbitrary monotone norm by interpreting it as collection of \wkm instances that we need to satisfy simultaneously. We will repeatedly solve such \wkm instances that are dynamically discovered during the course of the algorithm. 

It will be convenient to use the notion of {subgradients.} For our purposes, it is sufficient to discuss subgradients in the context of a monotone norm $f\colon \mathbb{R}^n \rightarrow \mathbb{R}$. We say that $\vc{g}$ is a \emph{subgradient} of $f$ at point $\vc{x}$ if $f(\vc{x}) = \vcg^{\intercal} \vc{x}$ and $f(\vc{y}) \ge  \vcg^\intercal \vc{y}$ for every $\vc{y} \in \bRna$. It is known that every monotone norm has a nonnegative subgradient $\vc{g}\ge 0$ at every point $\vc{x}\ge 0$. Checking whether a vector $\vc{g}$ is a subgradient at $\vc{x}$ and finding a subgradient at $\vc{x}$ can be formulated as convex optimization problems, hence can be (approximately) solved using the ellipsoid method if $f$ can be efficiently computed \cite{GLS}.

 Suppose that we have a current solution $X$ and let $\vc{x}\in \bRp$ be the vector representing the distances of the points in $P$ to $X$. Suppose that $X$ is not (approximately) optimal: $f(\vc{x})>(1+\epsilon)\opt$. Let us compute a sugradient $\vc{g}$ of $f$ at $\vc{x}$; we have  $\vcg^{\intercal} \vc{x} =f(x)>(1+\epsilon)\opt$ and  $\vcg^{\intercal} \vc{y}\le f(\vc{y}) = \opt$ for the optimum solution $\vc{y}$. That is, $\vcg^{\intercal} \vc{x}\le \opt$ is a linear constraint satisfied by the optimum solution and violated by the current solution. Then defining the weights $w(p)$ based on the coordinates of $\vc{g}$ gives an instance of \wkm, with $\sum_{p}w(p)\delta(p,X)>(1+\epsilon)\opt$ for the current solution $X$. Now we can proceed as above for the \wkm problem: we randomly choose a point $p$ and cluster $\iletter$, introduce a new request into $Q_{\iletter}$, find a new center~$x_\iletter$, etc., until we arrive to a solution $X$ with $\sum_{p}w(p)\delta(p,X)\le (1+\epsilon)\opt$. If this new solution $X$ is still nonoptimal for the original norm problem, that is, $f(\vc{x})>(1+\epsilon)\opt$, then we can again compute a subgradient, find a violated linear constraint (possibly the same as in the previous step). We repeat this until we find a solution with $f(\vc{x})\le (1+\epsilon)\opt$.

Defining the upper bounds and bootstrapping the algorithm  with a solution approximately satisfying the upper bounds were crucial for the analysis of the \wkm algorithm. For general norms, we can again define the upper bounds once we have the weights $w$ based on the violated linear constraint $\vcg^{\intercal} \vc{x}\le \opt$. However, these upper bounds would not be useful for the analysis, as they would depend on the violated linear constraint, hence would change during the algorithm.

Intuitively, we can see the constraint  $f(\vc{x})\le \opt$ as an infinite number of \wkm instances, corresponding to the linear constraints $\vcg^{\intercal} \vc{x}\le \opt$ for \emph{every} subgradient~$\vc{g}$ of~$f$. We would like to define $u(p)$ to be the smallest possible upper bound that can be assigned to $p$ among all of these infinitely many \wkm instances. Determining this value seems to be a difficult task, but actually the answer is very simple. Recall that $u(p)$ was defined as twice the smallest~$r$ such that $\ball(p,r)$ contains total weight at least $\opt/r$. Thus to define the upper bound $u(p)$, we need to know what the maximum weight of the points in $\ball(p,r)$ can be among the infinitely many instances corresponding to all the subgradients. Let $\vc{b}$ be the characteristic vector of $\ball(p,r)$ (i.e., every coordinate is 1 or 0, depending on whether a point is in or not in the ball). Then the question is to determine the maximum of $\vcg^{\intercal} \vc{b}$ among all subgradients $\vc{g}$. It is easy to see that this maximum is exactly $f(\vc{b})$: if $\vc{g}$ is a subgradient at $\vc{b}$, then  
$\vcg^{\intercal} \vc{b}=f(\vc{b})$; if $\vc{g}$ is a subgradient at an arbitrary point $\vc{y}$, then $\vcg^{\intercal} \vc{b}\le f(\vc{b})$. Thus we can determine the maximum weight of any ball and define the upper bounds accordingly. With these definitions, the analysis of the \wkm algorithm go through for general mononote norms. The two main properties of the upper bounds remain valid: (1) the upper bounds are satisfied by the optimum solution and (2) we can restrict our random choice of $p$ to points where the distance to the solution is not much smaller than $u(p)$.

In summary, the final algorithm consists of the following steps (see Figure~\ref{fig:algomain}). First we compute the upper bounds $u(p)$ and greedily find a 3-approximate solution satisfying these constraints. Then we repeat the following steps until we reach a solution $X$  for which the distance vector $\vc{x}$ satisfies $f(\vc{x})\le (1+\epsilon)\opt$. We compute a subgradient $\vc{g}$ of $f$ at $\vc{x}$ to obtain a violated linear constraint  $\vcg^{\intercal} \vc{x}\le \opt$. We randomly choose a point $p$ (according to the distribution described above) and require that $p$ be at most distance $\delta(p,X)/(1+\epsilon/3)$ from the solution, that is, we obtain a violated distance constraint. Then we randomly choose a cluster $\iletter\in [k]$ and require that this distance constraint be satisfied by center $x_\iletter$. Thus we put the request $(p,\delta(p,X)/(1+\epsilon/3)$ into $Q_{\iletter}$  find a new $x_\iletter$ that satisfy the cluster constraints imposed by the requests in $Q_{\iletter}$, if possible. We repeat these steps until we arrive to a solution $X$ with distance vector $\vc{x}$ satisfying $f(\vc{x})\le (1+\epsilon)\opt$. Our analysis shows that each step is consistent with a hypothetical optimum solution $O$ with probability $\Omega(\epsilon/k)$. Moreover, if \scatterdim is bounded, then the algorithm has to find a solution or fail after a bounded number of iterations.

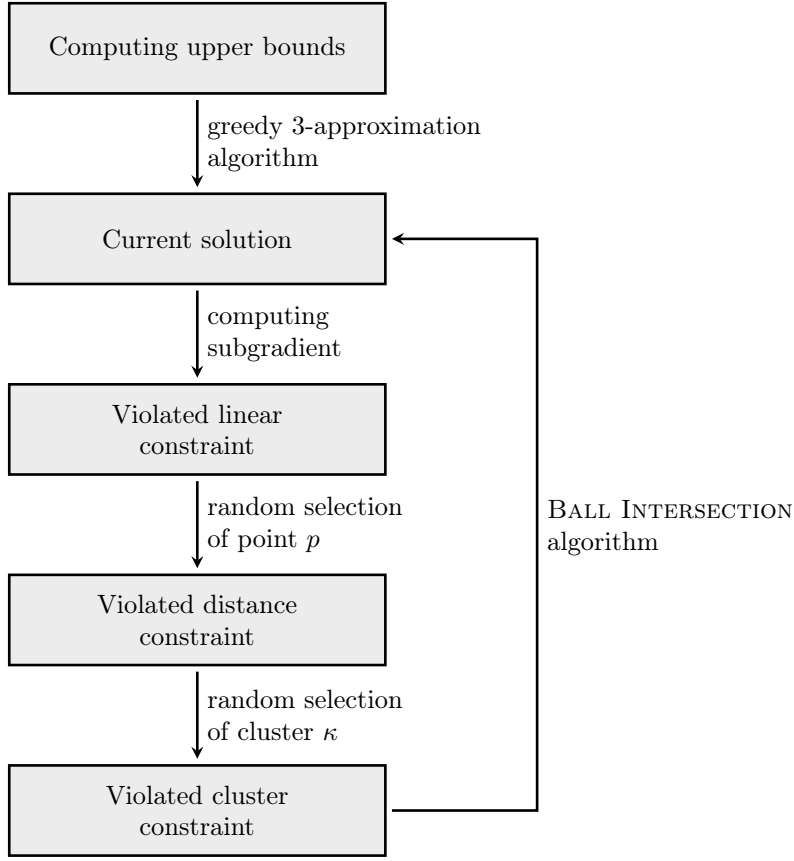
\begin{figure}[t]

\begin{center}
\begin{tikzpicture}[every node/.style={rectangle, line width=1pt,  minimum height=1.2cm},align=center]

  \small
\node (box1) [minimum width=5cm,draw,fill=gray!15] {Computing upper bounds};
\node (box2) [draw,minimum width=5cm,below=1.3cm of box1,fill=gray!15] {Current solution};
\node (box3) [draw,minimum width=5cm,below=1.3cm of box2,fill=gray!15] {Violated linear \\ constraint};
\node (box4) [draw,minimum width=5cm,below=1.3cm of box3,fill=gray!15] {Violated distance\\constraint};
\node (box5) [draw,minimum width=5cm,below=1.3cm of box4,fill=gray!15] {Violated cluster\\ constraint};
\node[align=left,above right=0cm and 2cm of box4] {\pc\\algorithm};
\draw [>=stealth, line width=1pt,->,shorten >=2pt,shorten <=2pt] (box1)  --  (box2) node[align=left,midway,right=0cm] {greedy 3-approximation\\algorithm};
\draw [>=stealth, line width=1pt,->,shorten >=2pt,shorten <=2pt] (box2) -- (box3) node[align=left,midway,right] {computing\\subgradient};
\draw [>=stealth, line width=1pt,->,shorten >=2pt,shorten <=2pt] (box3) -- (box4) node[align=left,midway,right] {random selection\\of point $p$};
\draw [>=stealth, line width=1pt,->,shorten >=2pt,shorten <=2pt] (box4) -- (box5) node[align=left,midway,right] {random selection\\of cluster $\iletter$};
\draw [>=stealth, line width=1pt,->,shorten >=2pt,shorten <=2pt] (box5.east)   -| ([xshift=2cm]box2.east)  -- (box2.east)  ;

\end{tikzpicture}
\caption{Overall structure of the main algorithm.}\label{fig:algomain}
\end{center}
\end{figure}

\paragraph{(Algorithmic) $\epsilon$-Scatter Dimension.} After the general algorithm capable of handling any monotone norm objective, our second main contribution is bounding the \scatterdim of various classes of metrics (Section~\ref{sec:applications}). In the interest of space, we do not go into the details of these (mostly combinatorial) proofs, but give only a brief overview.
\begin{itemize}
    \item \textbf{Bounded Doubling Dimension.} As outlined in the introduction, the set of points as well as the set of centers in an \scattering both form an $\epsilon$-packing of a unit ball implying that any metric of doubling dimension $d$ has \scatterdim $(\nicefrac{1}{\epsilon})^{\OO(d)}$. See Theorem~\ref{thm:doubling-metric}.

\item \textbf{Bounded-Treewidth Graph Metrics.} The \scatterdim bound for metrics defined by the shortest path metric of bounded-treewidth graphs is obtained by a delicate combinatorial proof that exploits both structure of the graph and properties of the \scattering. The bound we obtain is $\tw^{{1/\epsilon}^{\OO(\tw)}}$ for graphs of treewidth $\tw$, that is, double exponential in $\tw$ for fixed $\epsilon$. It remains is an interesting open question if this bound can be improved.

\item \textbf{Planar Graph Metrics.}
As outlined in the introduction, we can employ a known metric embedding result  to reduce the problem of bounding the \scatterdim of planar graphs to bounding the \scatterdim of bounded-treewidth graphs. In particular, the result by Fox-Epstein, Klein, and Schild \cite{fox2019embedding} provides an (approximate) metric embedding of planar metrics into low-treewidth metrics, which can be used to obtain a $2^{2^{\poly(1/\epsilon)}}$ bound on the \scatterdim of planar graph metrics.

\item \textbf{Continuous High-Dimensional Euclidean Space.}
As mentioned in the introduction, the high-dimensional Euclidean space does not have bounded \scatterdim. However, in the continuous Euclidean space, where any point of the space can be a center, we can bound the \emph{\algscatterdim}. Towards this, we replace the center player by an algorithmic ``player'' applying the algorithm by Kumar and Yildirim~\cite{KY09} for \textsc{Weighted 1-Center}. To achieve bounded \algscatterdim, this algorithm would require, however, a bounded aspect ratio of the radii in the input requests. We therefore prove an aspect-ratio condition (which holds even for general metrics) implying that it is sufficient for the algorithm to handle instances with aspect-ratio $\OO(\nicefrac{1}{\epsilon})$. We combine this result with the algorithm by Kumar and Yildirim to prove bounded \algscatterdim for continuous high-dimensional Euclidean space, that is, Theorem~\ref{thm:scatter-cont-eucl}.
\end{itemize}

\section{Preliminaries}\label{sec:prelims}
\paragraph{Classes of Metric Clustering Spaces.}
A \emph{metric clustering space} (or metric space for brevity) is a triple $M=(P,F,\delta)$ where $P$ is a finite set of $n$  \emph{data points}, $F$ is a (possibly infinite) set of potential locations of cluster centers, and $\delta$ is a metric on $P\cup F$. 
Sets $P$ and $F$ are not necessarily disjoint. (For example, it is natural for clustering problems to have $P=F$ or $P\subseteq F$.)
Given any point $u\in P\cup F$ in the metric space and a radius $r\in\mathbb{R}_{+}$, we denote by 
$\ball_{\delta}(u,r) = \{\,v\in P\cup F\mid \delta(u,v)\leq r\,\}$ the ball of radius $r$ centered around $u$. We drop the subscript $\delta$ if the distance function is clear from the context.

By $|M|$ we denote the space needed to represent the metric space $M$ in the memory. If $M$ is finite then $|M|$ is polynomial in $|F|$, $|P|$ and the space needed for storing a point and a center, respectively. If $F$ is infinite (for example, in the continuous Euclidean setting, $F={\mathbb R}^d$), $|M|$ is polynomial in~$|P|$ and the space of storing a point.

A \emph{class} $\cM$ of metric spaces is a (infinite) set of metric spaces. This paper focuses on metric classes that are closed under scaling distances by a constant.
We consider the following classes of metric clustering spaces: 

\begin{itemize}
    \item {\bf Graph Metric:} In the case of graph metric, we are given a (weighted) graph $G=(V,E)$ and the metric $\delta_G$ on $V$ as the shortest path metric, i.e., $\delta_G(u,v)$ is the shortest distance of a path connecting $u$ and $v$. The clustering space $(P, F, \delta_G)$ is given such that $P, F \subseteq V$. 

    \item {\bf Continuous Euclidean Spaces:} In this case, we are allowed to choose centers from the (high-dimensional) continuous Euclidean space $F = {\mathbb R}^d$. The set $P\subsetneq \bRd$ is a finite set of points.

    \item {\bf Doubling Metric:}  The {\it doubling dimension} of a metric space $(X,\delta)$, denoted as $d$, is the smallest $m>0$ such that every ball of radius $r$ in the metric can be covered by $2^m$ balls of radius $\frac{r}{2}$. 
 Note that a $d$-dimensional Euclidean metric has doubling dimension~$\OO(d)$.
\end{itemize}

\paragraph{Treewidth.} A {\em tree decomposition} of a graph $G$ is a pair $(T, \beta)$ where $T$ is a tree, $\beta\colon V(T) \to 2^{V(G)}$, $V(T)$ and $V(G)$ denote the vertices of the tree $T$ and $G$ respectively, with the following properties.

\begin{enumerate}
    \item For each $v\in V(G)$, there exists $t \in V(T)$ such that $v \in \beta(t)$,
    \item For each $(u,v) \in E(G)$, there exists $t \in V(T)$ such that $u,v \in \beta(t)$, and
    \item For each $v \in V(T)$, the subgraph induced by $T$ on $\{t: v \in \beta(t)\}$, is connected.
\end{enumerate}

The {\em width} of the tree decomposition $(T,\beta)$ is $\max_{t \in V(T)} |\beta(t)|-1$. The {\em treewidth} of a graph $G$ is the minimum width over all tree decompositions of $G$.

\paragraph*{Subgradients of Norms.}
We state definitions and summarize basic facts about subgradients of norms that we will use throughout the paper.
\begin{fact}
Any norm is a convex function.
\end{fact}

\begin{definition}[Subgradient]
    A subgradient of a convex function  $f\colon \mathbb{R}^n \rightarrow \mathbb{R}$ at any point $\vc{x} \in \bRna$ is any $\vcg \in \mathbb{R}^n$ such that the following holds for every $\vc{y} \in \bRna$
    \[
        f(\vc{y}) \ge f(\vc{x}) + \vcg^\intercal (\vc{y} - \vc{x});
    \]
    we denote by $\grad(\vc{x})$ the set of subgradients of $f$ at $\vc{x}$.
\end{definition}

The following fact summarizes various useful properties of subgradients specialized to norm functions.  Because we are apply norm objectives exclusively to non-negative distance vectors, we call (slightly abusing terminology) a restriction of a norm to $\bRn$ a norm as well.
\begin{fact}[\cite{chakrabarty2019approximation}] \label{fact:chkbty}
Let  $f\colon \bRn \rightarrow \bR_{\geq 0}$ be a norm and $\vc{x} \in \bRn$.
    If $\vcg$ is a subgradient of $f$ at $\vc{x}$, then $f(x) = \vcg^{\intercal} \vc{x}$ and $f(y) \ge  \vcg^\intercal \vc{y}$ for all $\vc{y} \in \bRn$. 
Further, if $f$ is monotone, there exists a subgradient $\vc{g}\in\grad(\vcx)$ such that $\vc{g} \ge 0$.
\end{fact}

The following observation is an immediate consequence of Fact~\ref{fact:chkbty}.
\begin{observation}\label{obs:norm-max-subgradients}
    Let $\grad= \bigcup_{\vc{y}\in\bRn} \grad(\vc{y})$ be the set of all subgradients of $f$. Then for any $\vc{x}\in\bRn$, we have that
    \begin{displaymath}
        f(\vc{x})=\max_{\vc{g}\in \partial f}\vc{g}^{\intercal}\vc{x}\,.
    \end{displaymath}
\end{observation}

\begin{definition}[$\epsilon$-Approximate Subgradient]\label{def:apx-subgradient}
Let $f\colon\bRn\rightarrow\bR_{\geq 0}$ be a norm and let $\epsilon>0$. We define the set $\gradapx(\vc{x})$ of \emph{$\epsilon$-approximate subgradients of $f$ at $\vc{x}$} to contain all $\vcg \in \bRn$ such that the following two conditions hold
\begin{enumerate}[(i)]
    \item $f(\vc{y})\geq \vcg^{\intercal} \vc{y}$ for each $\vc{y}\in\bRn$, and
    \item $f(\vc{x})\leq (1+\epsilon)\vcg^{\intercal}\vc{x}$\,.
\end{enumerate}
\end{definition}

It is known that approximate subgradients of convex functions can be computed efficiently
via an (approximate) value oracle for the function through reductions shown by Gr\"{o}tschel, Lovasz and Schrijver in their classic book~\cite{GLS}. While the reduction in~\cite{GLS} appears to take at least $\Omega(n^{10})$ calls to the
oracle, there exist faster methods assuming additional properties of the convex function, for example, see \cite{lovasz2006fast, lee2018efficient}. Specifically for $\ell_p$ norms, closed formulas describing the sets of subgradients are known and used in practice.


\paragraph*{Some Terminology and Notation.}
Let $M=(P,F,\delta)$ be a clustering space on $n=|P|$ data points. Let $\vc{b}\in\bRp$ be an $n$-dimensional vector. We interpret $\vc{b}$ as assigning each point $p\in P$ a non-negative value denoted $b(p)$. That is, $\vc{b}=(b(p))_{p\in P}$. For example, given a subset $X\subseteq F$ of centers, we define the \emph{distance vector} $\vcd(P,X)=(\delta(p,X))_{p\in P}$.  If $B\subseteq P$ is a subset of points then $\vc{1}_B\in\{0,1\}^P$ denotes the \emph{characteristic vector of $B$}, that is, it assigns value $1$ to any $b\in B$ and $0$ to any $p\in P\setminus B$. If $p\in P$ and $\alpha\geq 0$ then we denote by $\vc{1}_{p,\alpha}$ the binary vector $\vc{1}_{\ball(p,\alpha)\cap P}$.

\section{\ScatterDim}\label{sec:scatter-dim}
In this section, we introduce the concept of \scatterdim formally, which plays a  central role in our algorithmic framework. The following definition is a formalization of the ``center-point game'' presented in the introduction.
\begin{definition}[\ScatterDim]\label{def:scattercomp}
We are given a class $\cM$ of finite metric spaces, a space $M=(P,F,\delta)$ in $\cM$, and some $\epsilon\in (0,1)$. 
An \emph{\scattering in $M$} is a sequence $(x_1,p_1)\dots,(x_{\ell},p_{\ell})$ of center-point pairs $x_i\in F$, $p_i\in P$, $i\in[\ell]$ such that
\begin{align*}
    \delta(x_i,p_j)&\leq 1 && \textnormal{ for all } 1\leq j < i\leq\ell & \textnormal{(covering)}\\
    \delta(x_i,p_i)&>1+\epsilon && \textnormal{ for all } i \in [\ell] & \textnormal{($\epsilon$-refutation)}
\end{align*}
The \emph{\scatterdim of $M$} is the maximum length of an \scattering in it. The \emph{\scatterdim of $\cM$} is the supremum of the \scatterdim over all $M\in\cM$.
\end{definition}

Note that for any \scattering $(x_1,p_1), \dots,(x_{\ell},p_{\ell})$, any subsequence $(x_{i_1},p_{i_1}),\dots,\allowbreak (x_{i_{\ell'}},p_{i_{\ell'}})$ where $i_1<\dots <i_{\ell'}$ and $\ell',i_j\in[\ell]$, $j\in[\ell']$ is an \scattering as well.

As described in Theorem~\ref{thm:main}, we show that bounded (algorithmic) \scatterdim is essentially sufficient to yield an EPAS for \knc in the respective metric space. In Section~\ref{sec:graph-metrics} we show that bounded treewidth and planar graph metrics, and bounded doubling metrics have bounded \scatterdim. This allows us to obtain EPASes in all these metrics. To handle high-dimensional Euclidean space, we resort to an algorithmic version of \scatterdim.

\paragraph*{Optimizing the Centering Strategy.} Recall the example from the introduction showing that the \scatterdim of the high-dimensional (continuous) Euclidean space $\mathbb{R}^d$ is be unbounded.  We constructed an \scattering $(x_1,p_1),\dots,\allowbreak(x_{d-1},p_{d-1})$ where $x_i$ is the $i$-th unit vector scaled by $\nicefrac{1}{\sqrt{2}}$ and where $p_i=-x_i$ for all $i \in [d-1]$. Note that in this example the unit ball around the origin contains all the points in the sequence. 
Hence, the above example would collapse if the center player would improve her strategy. This motivates us to consider a variant where we replace the center player with an algorithm that computes centers more prudently. Further, employing algorithms allows us also to handle infinite spaces.

\paragraph*{\pc Problem and Algorithm.} Towards this, we formalize the algorithmic problem the center player has to solve. We adopt and generalize a dual interpretation of the center-point game in which the center player is trying to find a center in the \emph{intersection} of all unit balls around the points played by the point player. In fact, we consider the more general setting of non-uniform balls where each point $p$ in the scattering has its own dedicated radius $r$.

Let $\cM$ be a class of metric spaces $(P,F,\delta)$ with possibly infinite center sets $F$. We define the following search problem.
\medskip

\defproblemout{\pc}{A metric space $M=(P,F,\delta)\in\cM$, a set finite set $Q\subsetneq P\times\bR_{+}$ of \emph{distance constraints}.}{A point $x\in F$ \emph{satisfying} all distance constraints, that is, $\delta(x,p)\leq r$ for each $(p,r)\in Q$, if such a point exists and ``fail'' otherwise.}
\medskip

For finite metric spaces, the \pc problem can be solved efficiently by exhaustively searching the center space $F$.
Unfortunately, we are not aware of exact algorithms for \pc for certain infinite metric spaces such as high-dimensional continuous Euclidean space. We therefore work with approximate algorithms. To define this formally, we say that a center $x\in F$ \emph{$\eta$-satisfies} the distance constraint $(p,r)\in Q$ for some error parameter $\eta>0$, if $\delta(x,p)\leq (1+\eta)r$. Let $\cC_{\cM}$ be a (deterministic) algorithm whose input is an instance of \pc and an error parameter $\eta>0$. The algorithm is called an \emph{approximate \pc algorithm} (or \pc algorithm for short) if it satisfies the following conditions.
\begin{enumerate}[(i)]
    \item The algorithm outputs a center that $\eta$-satisfies all distance constraints or it fails.
    \item If there exists a center satisfying all points distance constraints exactly, then the algorithm does not fail.
    \item The running time of $\cC_{\cM}$ is $\poly(|M|,1/\eta)$.
\end{enumerate}
We remark that there is an approximate \pc algorithm for high-dimensional Euclidean space~\cite{KY09}, which we employ in Section~\ref{sec:high-dim-cont} to prove bounded \algscatterdim of this metric.

\paragraph*{\AlgScatterDim.} The definition of \emph{\algscatterdim} is based on the notion of \emph{\algscattering}, which is a variant of \scattering: Centers are chosen via an (approximate) \pc algorithm $\cC_{\cM}$ rather than by an adversarial center-player. Intuitively, we maintain a dynamic instance of \pc that is augmented by adding distance constraints $(p,r)$ one by one. In the context of \algscattering, we call the distance constraints $(p,r)$ \emph{requests}, which are satisfied by the \pc algorithm sequentially.

\begin{definition}[\AlgScatterDim]\label{def:algscattercomp}
Let $\cM$ be a class of metric spaces with \pc algorithm $\cC_{\cM}$, let $M=(P,F,\delta)$ be a metric in $\cM$, and let $\epsilon\in (0,1)$ 
Moreover, let $p_i\in P$, $x_i\in F$, and $r_i\in\mathbb{R}_+$ for each $i\in [\ell]$ where $\ell$ is a positive integer. The sequence $(x_1,p_1,r_1),\dots,(x_\ell,p_\ell,r_{\ell})$ is called an \emph{\algscattering} if the following two  conditions hold. \begin{enumerate}[(i)]
    \item We have $x_i=\cC_{\cM}(M,\{(p_1,r_1),\dots,(p_{i-1},r_{i-1})\},\nicefrac{\epsilon}{2})$ for each $2\leq i\leq\ell$. (There is no requirement regarding the first center $x_1$ in the sequence.) 
    \item Moreover, $\delta(x_i,p_i)>(1+\epsilon)r_i$ for each $i\in[\ell]$.
\end{enumerate}
We say that $\cM$ has \emph{\algscatterdim[$(\epsilon,\cC_{\cM})$] $\lambda_{\cM}(\epsilon)$} if any \algscattering contains at most $\lambda_{\cM}(\epsilon)$ many triples with the same radius value.  The \emph{\algscatterdim} of~$\cM$ is the minimum \algscatterdim[$(\epsilon,\cC_{\cM})$] over any \pc algorithm $\cC_{\cM}$ for $\cM$.
\end{definition}

When the family $\cM$ is clear from the context we drop the subscript $\cM$ from $\lambda_{\cM}(\epsilon)$ and $\cC_{\cM}$.
Note that, in contrast to the \scatterdim, for \algscatterdim we demand that the number of triples per radius value be bounded rather than the total length of the sequence. In fact, this stronger requirement would not hold for high-dimensional Euclidean spaces whereas the weaker (algorithmic) requirement turns out to be sufficient for our results. Another noteworthy difference is that a subsequence of an \algscattering is not necessarily a \algscattering itself because it may not be consistent with the behavior of algorithm $\cC_{\cM}$.

\paragraph*{Relation Between Algorithmic and non-Algorithmic \ScatterDim.}
The following lemma shows that the \algscatterdim indeed generalizes the \scatterdim for finite metric spaces. 

\begin{lemma}\label{lem:scatterd-dim-bound-by-alg}
Any class of finite, explicitly given, metric spaces with \scatterdim $\lambda(\epsilon)$ has also \algscatterdim $\lambda(\epsilon)$.
\end{lemma}
\begin{proof}

Let $M=(P,F,\delta)$ be a metric space in the given class along with a set $Q$ of distance constraints. Our \pc algorithm exhaustively searches $F$ to find a center $x$ satisfying all distance constraints \emph{exactly}. If no such point exists the algorithm fails. Let $\cC$ denote this algorithm. Consider any \algscattering. Notice that any sub-sequence of triples with the same radius value forms an \scattering. Hence the sequence contains at most $\lambda(\epsilon)$ many triples for any radius value.
\end{proof}

\paragraph*{Aspect-Ratio Lemma for \AlgScatterDim.}
The following is a handy consequence of bounded \algscatterdim that we use in proving our result. It strengthens the properties of an \algscattering by bounding the number of triples whose radii lie in an interval of bounded aspect-ratio (rather than bounding the number of triples with the same radius value).
\begin{lemma}\label{lem:slice-alg-scatter-dim}
    Let $\cM$ be a class of metric spaces of \algscatterdim $\lambda(\epsilon)$. Then there exists a \pc algorithm $\cC_{\cM}$ with the following property. Given $\epsilon\in (0,1)$, $a>0$, and $\tau\geq 2$, any \algscattering contains $\OO(\lambda(\nicefrac{\epsilon}{2})(\log \tau)/\epsilon)$ many triples whose radii lie in the interval $[a,\tau a]$. 
\end{lemma}
\begin{proof}
    It suffices to show the weaker claim that the number of requests in the interval $[a,(1+\nicefrac{\epsilon}{100})a]$ is at most $2\lambda(\nicefrac{\epsilon}{2})$. This claim implies the lemma because the interval $[a,\tau a]$ can be covered with $\OO((\log\tau)/\epsilon)$ many intervals of the form $[(1+\nicefrac{\epsilon}{100})^j,(1+\nicefrac{\epsilon}{100})^{j+1}]$, $j\in\bZ$.
    
    Let $\cA_{\cM}$ be an \pc algorithm such that the \algscatterdim[$(\epsilon,\cA_{\cM})$] is $\lambda(\epsilon)$. Let $\eta\in (0,1)$ be the input error parameter. Consider the \pc algorithm $\cC_{\cM}$ that works as follows. For any of the input requests $(p,r)$ we round $r$ to $r'$, which is the smallest power of $1+\nicefrac{\eta}{50}$ larger than $r$. We then invoke $\cA_{\cM}$ on the rounded requests with error parameter $\eta/2$ and output the center returned by $\cA_{\cM}$. Clearly, this algorithm is an $(1+\eta)$-approximate \pc algorithm (for the original requests). 

    Consider any algorithmic \algscattering $(x_1,p_1,r_1),\dots,(x_\ell,p_\ell,r_\ell)$. Let $r_i'$, $i\in[\ell]$ be the rounded radii computed by $\cC_{\cM}$. Let $\epsilon'=\epsilon/2$. Let $1\leq j<i\leq\ell$. We have $\delta(x_i,p_i)>(1+\epsilon)r_i\geq (1+\epsilon)/(1+\epsilon/100)r_i'\geq (1+\epsilon')r_i'$. Moreover, we have $\delta(x_i,p_j)\leq (1+\epsilon'/2)r_i'$. Hence, the sequence $(x_1,p_1,r_1'),\dots,(x_\ell,p_\ell,r_\ell')$ is an algorithmic \algscattering[$(\cC_{\cM},\epsilon')$]. The radii in the requests $(p_i,r_i)$, $i\in[\ell]$ that lie in the interval $[a,(1+\nicefrac{\epsilon}{100})a]$ are rounded by $\cC_{\cM}$ to at most two distinct radius values because $\cC_{\cM}$ is invoked with error parameter $\eta=\epsilon/2$. Hence the (unrounded) sequence contains at most $2\lambda(\epsilon')=2\lambda(\epsilon/2)$ many triples with radii in the interval $[a,(1+\nicefrac{\epsilon}{100})a]$. This completes the proof of the claim and therefore of the lemma.
\end{proof}

\section{Framework for Efficient Parameterized Approximation Schemes}
\label{sec:framework}

\paragraph*{Main Result.}
We are now ready to state our main result. In the remainder of this section, we prove the following theorem, restated from the introduction. In Section~\ref{sec:algorithm}, we describe the EPAS and give some intuition. In Section~\ref{sec:analysis}, we give a full, technical analysis.

\thmmain*

\subsection{Algorithm}\label{sec:algorithm}
Our algorithm is stated formally in Algorithm~\ref{alg:framework-knc}.
We informally summarize the key steps of our algorithm, which we also outlined partially in the technical overview. We also give some intuition of the analysis.

Using standard enumeration techniques, we assume that we know (a sufficiently exact approximation of) the optimum objective function value $\opt$. Our goal is to satisfy the \emph{convex constraint} $f(\vc{x})\leq(1+\epsilon)\opt$ imposed on the distance vector $\vcx\in\bRp$ (which represents the distance vector $\vcd(P,X)$ induced by the feasible solution $X\subseteq F$).
By Observation~\ref{obs:norm-max-subgradients}, this constraint is equivalent to (infinitely many) \emph{linear constraints} $\vcw^{\intercal}\vcx\leq(1+\epsilon)\opt$ where $\vcw\in\grad$ is any subgradient of $f$.

To illustrate the main idea, we first describe a highly simplified, but failed attempt. We consider in each iteration of the while-loop (lines~\ref{step:loop}--\ref{step:loop-end}) a candidate solution $X$. If $f(\vc{x})\leq (1+\epsilon)\opt$, then we are done. Otherwise, we compute an ($\nicefrac{\epsilon}{10}$-approximate) subgradient $\vcw$ of $f$ at $\vcx$ in line~\ref{step:apxsg}. Since $\vcw^{\intercal}\vcx=f(\vcx)>(1+\epsilon)\opt$, this constitutes a \emph{violated} linear constraint. Consider sampling a point $p\in P$ with probability proportional to its contribution $w(p)\delta(p,X)$ to the objective $f(\vc{x})=\vcw^{\intercal}\vcx$ (line~\ref{alg:sample-witness}). An averaging argument shows that with probability $\Omega(\epsilon)$, the sampled point $p$ satisfies $\delta(p,X)>(1+\nicefrac{\epsilon}{3})\delta(p,O)$ for some fixed hypothetical optimum solution $O$. In this event, we identified a violated \emph{distance constraint}, and call $p$ an \emph{$\nicefrac{\epsilon}{3}$-witness} for $X$. We assign $p$ to a cluster $\iletter\in[k]$ picked uniformly at random, which equals the correct cluster of $p$ in $O$ with probability~$\nicefrac{1}{k}$. Assuming that both events occur, this allows us to add the \emph{request} $(p,r)$ with radius value $r=\delta(p,X)/(1+\nicefrac{\epsilon}{3})$ to the \emph{cluster constraint} $Q_\iletter$ imposed on the cluster with index $\iletter$. (See lines~\ref{alg:augment-cluster} and~\ref{alg:update-center}.) Here, we refer to the set $Q_\iletter$ of requests for cluster $\iletter$ as cluster constraint of $\iletter$.

Fix cluster index $\iletter\in[k]$. Let $(p_\iletter^{(1)},r_\iletter^{(1)}),\dots,(p_\iletter^{(\ell)},r_\iletter^{(\ell)})$ be the sequence of requests added to the cluster constraint associated with cluster $\iletter$. Let $x_{\iletter}^{(i)}$, $i\in [\ell]$ be the center of cluster $\iletter$ just before adding the request $(p_\iletter^{(i)},r_\iletter^{(i)})$ to $Q_\iletter$. The key observation is that the sequence of triples $(x_{\iletter}^{(1)},p_\iletter^{(1)},r_\iletter^{(1)}),\dots,(x_{\iletter}^{(\ell)},p_\iletter^{(\ell)},r_\iletter^{(\ell)})$ forms an algorithmic \scattering. We would like to argue that the length of this sequence is bounded because the \algscatterdim is bounded. Unfortunately, the scatter dimension bounds only the number of triples per radius value but not the overall length of the sequence.

To address this issue, we compute in line~\ref{step:compute-initial-upper-bounds} an initial upper bound $u(p)$ on the radius of any point $p\in P$. We (approximately) satisfy these initial distance constraints for all points in a greedy pre-processing step (see lines~\ref{step:start-preprocess}--\ref{step:end-preprocess}). We maintain the distance constraints during the main phase by adding them as initial requests (see line~\ref{step:initial-requests}). The upper bound $u(p)$ is a rough estimate of the smallest radius $r$ that may be imposed on $p$ as part of any request $(p,r)$. We modify the sampling process in the main phase (see line~\ref{alg:sample-witness}) to sample only from a subset of points whose distance to $X$ is not much smaller than their initial upper bound $u(p)$. We show via a careful argument that every request $(p,r)$ we make is consistent with $O$ with probability $\Omega(\nicefrac{\epsilon}{k})$. We argue, moreover, that all radii of requests made for a particular cluster are within a factor $\OO(\nicefrac{k}{\epsilon^2})$ of each other. The initial upper bounds are computed by detecting ``dense'' balls (line~\ref{step:compute-initial-upper-bounds}) in the input instance in the sense that they would receive high weight by some subgradient of the objective norm and would therefore require that any near-optimal solution must place a center in the vicinity of that dense ball.

\medskip
\begin{algorithm}[H]

\SetKwFunction{algo}{algo}
\caption{Framework for \knc}\label{alg:framework-knc}
  \KwData{Instance $\cI=((P,F,\delta), k, f\colon\bRp \rightarrow \bR_{\ge 0})$ of \knc, error parameter $\epsilon\in (0,1)$, $\opt>0$, \pc algorithm $\cC$ according to Lemma~\ref{lem:slice-alg-scatter-dim}}
  \KwResult{Solution $X$ of cost at most $(1+\epsilon)\opt$ if solution of cost at most $\opt$ exists}
  For each $p \in P$, compute $u(p)= \min\{\,\alpha > 0 \mid  f(\vc{1}_{p,\nicefrac{\alpha}{3}})\ge 3\opt/\alpha\,\}$\;\label{step:compute-initial-upper-bounds}
  Sort $P$ in non-decreasing order of $u(p)$\;\label{step:start-preprocess}
  Mark $p_i\in P$ if $\ball(p_i,u(p_i))$ is disjoint from  $\ball(p_j,u(p_j)$ for every $j<i$\;\label{step:mark-points}
  Let $p^{(1)}$, $\dots$, $p^{(k')}$ be the marked points.\tcp*{Lemma~\ref{lem:anchors-satisfied} shows that $k'\leq k$}

  Let $Q_\iletter=\{(p^{(\iletter)},u(p^{(\iletter)}))\}$ for all $\iletter \in [k']$\;\label{step:initial-requests}
  Let $Q_\iletter=\emptyset$ for all $\iletter$ with $k'<\iletter\leq k$\;
  Let $X=(x_1,\dots,x_k)$ be any set of centers where $x_\iletter$ satisfies the requests in $Q_\iletter$\;\label{step:end-preprocess}

  \While{$f(\vcd(P,X))>(1+\epsilon)\opt$}{\label{step:loop}
           $\vcw \gets \text{$\nicefrac{\epsilon}{10}$-subgradient of $f$ at $\vcd(P,X)$}$\;\label{step:apxsg}
     $A\gets\left\{\,p\in P\mid \delta(p,X)\geq \frac{\epsilon u(p)}{1000k}\,\right\}$\label{alg:admissable-set}\;
     Sample an element $p\in A$ where $\prob{p=a}=\frac{w(a)\delta(a,X)}{\sum_{b\in A}w(b)\delta(b,X)}$ for any $a\in A$\;\label{alg:sample-witness}
     Pick cluster $\iletter\in[k]$ for $p$ uniformly at random\;
     $Q_\iletter\gets Q_\iletter\cup \{(p,\delta(p,X)/(1+\nicefrac{\epsilon}{3}))\}$\;\label{alg:augment-cluster}
     $x_\iletter\gets$ $\cC(Q_\iletter,\nicefrac{\epsilon}{10})$
    \lIf{no $x_i$ was found}{
         \textbf{fail}
     }\label{alg:update-center}
  }
  \label{step:loop-end}
  \Return{$X$\;}
   
\end{algorithm}

\subsection{Analysis}\label{sec:analysis}
\paragraph*{Overview.} The analysis consists of establishing the following three facts. First, if the algorithm terminates without failure, it computes a $(1+\epsilon)$-approximation. Second, the algorithm terminates---with or without failure---after a number of iterations that depends on $k$ and $\epsilon$ only. Third, the algorithm does not fail with a probability that depends only on $k$ and $\epsilon$ as well.

The first step of the analysis follows immediately from the stopping criterion (line~\ref{step:loop}) of the while loop.
\begin{observation}[Correctness]\label{obs:apx-guarantee}
    If the algorithm terminates successfully (that is, without failure), then it outputs a $(1+\epsilon)$-approximate solution.
\end{observation}

The second step of the analysis is summarized in the following lemma, which we prove in Subsection~\ref{sec:bounding-iterations}.
\begin{lemma}[Runtime bound]\label{lem:termination}
    The algorithm terminates after $\OO\left(\frac{k(\log\nicefrac{k}{\epsilon})\lambda(\nicefrac{\epsilon}{10})}{\epsilon}\right)$ iterations---with or without failure.
\end{lemma}

With these two insights at hand, we are left with the third step summarized by the following lemma, which we prove in Subsection~\ref{sec:bounding-prob}.
\begin{lemma}[Probability bound]\label{lem:success-prob}
    The algorithm terminates successfully (that is, without failure) with probability $ \exp{\left(-\widetilde{\OO}\left(\frac{k\lambda(\nicefrac{\epsilon}{10})}{\epsilon}\right)\right)}$.
\end{lemma}

We repeat the algorithm $\exp{\left(\widetilde{\OO}\left(\frac{k\lambda(\nicefrac{\epsilon}{10})}{\epsilon}\right)\right)}$ many times and hence succeed with high probability by Lemma~\ref{lem:success-prob}.

 The remainder of this section is dedicated to proving Lemmas~\ref{lem:termination} and~\ref{lem:success-prob},  thereby completing proof of the main Theorem~\ref{thm:main}.

 \subsubsection{Bounding the Number of Iterations}\label{sec:bounding-iterations}
In this subsection, we prove Lemma~\ref{lem:termination}. The proof consists in three steps. First, we argue that the initial upper distance bounds $u(p)$ that we compute for each point $p\in P$ are (i) consistent with any optimum solution (Lemma~\ref{lem:anchors}), and (ii) approximately satisfied throughout the algorithm (Lemma~\ref{lem:anchors-satisfied}). Second, we establish that the radii in the requests made for any particular cluster are within a bounded factor (aspect ratio) of each other (Lemma~\ref{lem:aspect-ratio}). The third step consists in proving that, for any particular cluster, the sequence of requests along with the corresponding centers constitute an algorithmic \algscattering of bounded aspect ratio. Hence we can use Lemma~\ref{lem:slice-alg-scatter-dim} to bound the length of the sequence and thus the number of iterations by a function of $k$ and $\epsilon$, thereby completing the proof of Lemma~\ref{lem:termination}.

\paragraph*{Initial Upper Bounds.} We first show that the initial upper bounds we calculate in the algorithm are conservative in the sense that they are also respected by an optimal solution.

\begin{lemma}\label{lem:anchors}
    If $O$ is an optimal solution then $\delta(p,O)\leq u(p)$ for any $p\in P$, where $u(p)$ is the initial upper bound computed in line~\ref{step:compute-initial-upper-bounds} of Algorithm~\ref{alg:framework-knc}.
\end{lemma}
\begin{proof}
    Let $\alpha=u(p)$. For the sake of a contradiction, assume that $\delta(p,O)>\alpha$. By triangle inequality, any point $p'\in\ball(p, \alpha/3)$ has distance at least $2\alpha/3$ to $O$. Hence we have $\vcd(P,O)\geq (2\alpha/3)\cdot \vc{1}_{p,\nicefrac{\alpha}{3}}$ and thus $f(\vcd(P,O))\geq f((2 \alpha/3)\cdot \vc{1}_{p,\nicefrac{\alpha}{3}})=(2 \alpha/3) f(\vc{1}_{p,\nicefrac{\alpha}{3}})\geq (2\alpha/3)\cdot\frac{3\opt}{\alpha}=2\opt$, which is a contradiction.
\end{proof}
 
The following lemma says that throughout the algorithm we approximately satisfy all upper bounds. We remark that the initialization (lines~\ref{step:start-preprocess}--\ref{step:end-preprocess}) as well as the analysis is a variant of Plesn\'\i k's algorithm~\cite{plesnik1987heuristic} for \pkc when applied to point set $P$ with radii $u(p)$, $p\in P$.

\begin{lemma}\label{lem:anchors-satisfied}
The number $k'$ of points marked in line~\ref{step:mark-points} in Algorithm~\ref{alg:framework-knc} is at most $k$. Moreover, at any time during the execution of the while loop (lines~\ref{step:loop}--\ref{step:loop-end}), we have that $\delta(p,X)\leq 4u(p)$. For any request $(p,r)$ added to some cluster constraint, we have $r\leq 4u(p)$.
\end{lemma}
\begin{proof}
    By Lemma~\ref{lem:anchors} each of the balls $\ball(p^{(\iletter)},u(p^{(\iletter)}))$ with marked $p^{(\iletter)}$, $\iletter\in [k']$ contains at least one point from some hypothetical optimum solution~$O$. On the other hand, these balls are pairwise disjoint by construction. Hence $k'\leq |O|\leq k$. This also implies that the algorithm can initialize $X=(x_1,\dots,x_k)$ in line~\ref{step:end-preprocess} with centers satisfying all initial cluster constraints. For example, it may pick the $k'$ centers in $F$ closest to $p^{(\iletter)}$, $\iletter\in[k']$ and $k-k'$ many additional arbitrary centers.
   
    Because these initial requests are never removed, they are passed to the \pc algorithm (with error parameter $\nicefrac{\epsilon}{10}$; see line~\ref{alg:update-center}) whenever we make a change in the respective cluster. Hence,  we have $\delta(p,X)\leq (1+\nicefrac{\epsilon}{10})u(p)\leq 3u(p)/2$ for any marked point $p$ throughout the execution of the while loop. For any point $p'$ \emph{not} marked, $\ball(p',u(p'))$ intersects $\ball(p,u(p))$ for some marked $p$. Because  the points are processed in line~\ref{step:mark-points} in non-decreasing order of $u(\cdot)$, we must have $u(p)\leq u(p')$. As argued before, $\ball(p,3u(p)/2)$ is guaranteed to contain a center in $X$ at any time during the while loop. This center has distance at most $u(p')+2\cdot 3u(p)/2\leq 4u(p')$ from $p'$ by triangle inequality. For the second claim, notice that $r<\delta(p,X)\leq 4u(p)$ at the time this request is processed in line~\ref{alg:update-center} for the first time.
\end{proof}

\paragraph*{Bounding the Aspect-Ratio of Requests.} The following lemma establishes that the radii of any two requests made for the same cluster are within a factor $\OO(\nicefrac{k}{\epsilon^2})$ from each other. The intuition is as follows. We ensure in the algorithm (see line~\ref{alg:admissable-set}) that we only sample points whose radii are within a factor $\OO(\nicefrac{k}{\epsilon})$ from $u(p)$. Assume that the radii, and thus the initial bounds $u(p)$, $u(p')$, in two request $(p,r)$, $(p',r')$ to the same cluster were very far from each other, say $r'\ll r$ and $u(p')\ll u(p)$. This would then imply that $p$ was already (essentially) within radius $r$ from some center before requesting $(p,r)$ since there must be a center within radius $4u(p')\ll \epsilon r/3$ from $p'$ by Lemma~\ref{lem:anchors-satisfied}. This contradicts the assumption that we requested $(p,r)$ in the first place.

\begin{lemma}\label{lem:aspect-ratio}
   Let $(p,r)$ and $(p',r')$ be requests added (in either order) to the same cluster constraint $Q_\iletter$, $\iletter\in[k]$ in line~\ref{alg:augment-cluster} of Algorithm~\ref{alg:framework-knc}. If $r'\leq \epsilon^2\cdot r/(10^4k)$ then the algorithm fails in line~\ref{alg:update-center} upon making the second of the two requests.
\end{lemma}
\begin{proof}
 Assume for the sake of a contradiction that the algorithm does not fail but finds a center $x_\iletter$ with $\delta(p,x_\iletter)\leq (1+\nicefrac{\epsilon}{10})r$ and $\delta(p',x_\iletter)\leq (1+\nicefrac{\epsilon}{10})r'$. Hence $\delta(p,p')\leq (1+\nicefrac{\epsilon}{10})(r+r')$ by triangle inequality. By Lemma~\ref{lem:anchors-satisfied}, we have $r\leq 4u(p)$ and $r'\leq 4u(p')$. Because we sample points from the set $A$ defined in line~\ref{alg:admissable-set}, we have $r\geq \epsilon u(p)/(200k)$ and $r'\geq\epsilon u(p')/(200k)$. 

Suppose $r'\leq \epsilon^2 r/(10^4k)$. At the time of adding $(p,r)$ to $Q_\iletter$ the current candidate solution $X$ satisfies $\delta(p',X)\leq 4u(p')\leq 1000kr'/\epsilon$ by Lemma~\ref{lem:anchors-satisfied}. Hence 
\begin{align*}
\delta(p,X) & \leq \delta(p,p')+\delta(p',X)\\
            & \leq (1+\nicefrac{\epsilon}{10})(r+r')+1000kr'/\epsilon\\
            & \leq (1+\nicefrac{\epsilon}{4})r\,.
\end{align*}
However, this is a contradiction because $\delta(p,X)=(1+\nicefrac{\epsilon}{3})r$ when requesting $(p,r)$ to $Q_\iletter$ as can be seen from line~\ref{alg:update-center}.
\end{proof}

\paragraph*{Leveraging Bounded Algorithmic \ScatterDim.} To complete the proof of Lemma~\ref{lem:termination}, we fix some cluster and consider the sequence of triples $(x,p,r)$ where $(p,r)$ is a request made for this cluster and where $x$ is the center of the cluster just before the request was made. We establish that this sequence constitutes an algorithmic \algscattering  and use Lemma~\ref{lem:aspect-ratio} to bound the aspect ratio of the radii in this sequence by $\OO(\nicefrac{k}{\epsilon^2})$. We complete the proof via the aspect-ratio lemma~\ref{lem:slice-alg-scatter-dim}.

\begin{proof}[Proof of Lemma~\ref{lem:termination}]
     Fix a cluster index $\iletter\in[k]$. Let  $(p_{\iletter}^{(1)},r_{\iletter}^{(1)}),\dots,(p_{\iletter}^{(\ell)},r_{\iletter}^{(\ell)})$ be the sequence of requests in the order in which they are added to $Q_\iletter$ in line~\ref{alg:augment-cluster}. For any $i\in[\ell]$, let $x_\iletter^{(i)}$ be the center of cluster $\iletter$ at the time just before requesting $(p_{\iletter}^{(i)},r_{\iletter}^{(i)})$. Since $r_{\iletter}^{(i)}=\delta(p_{\iletter}^{(i)},X)/(1+\nicefrac{\epsilon}{3})\leq \delta(p_{\iletter}^{(i)},x_\iletter^{(i)})/(1+\nicefrac{\epsilon}{3})$ and since $x_{\iletter}^{(i)}$ is computed by invoking $\cC$ on $\{(p_{\iletter}^{(1)},r_{\iletter}^{(1)}),\dots,(p_{\iletter}^{(i-1)},r_{\iletter}^{(i-1)})\}$ and error parameter $\nicefrac{\epsilon}{10}$, the sequence $(x_\iletter^{(1)},p_{\iletter}^{(1)},r_{\iletter}^{(1)}),\dots,(x_\iletter^{(1)},p_{\iletter}^{(\ell)},r_{\iletter}^{(\ell)})$    
    is an algorithmic \algscattering[$\nicefrac{\epsilon}{5}$].

    By Lemma~\ref{lem:aspect-ratio}, $r_{\iletter}^{(i)}\in R_\iletter=\left[r_{\min},\frac{10^4kr_{\min}}{\epsilon^2}\right]$ for every $i\in[\ell]$ where $r_{\min}$ denotes the smallest radius in any request for cluster $\iletter$. Applying Lemma~\ref{lem:slice-alg-scatter-dim} to the interval $R_\iletter$, the length of the sequence is $\OO((\log \nicefrac{k}{\epsilon})\lambda(\nicefrac{\epsilon}{10})/\epsilon)$. Since our algorithm adds in each iteration one request to some cluster constraint, the overall number of iterations is $\OO(k(\log\nicefrac{k}{\epsilon})\lambda(\nicefrac{\epsilon}{10})/\epsilon)$.
\end{proof}

\subsubsection{Bounding the Success Probability}\label{sec:bounding-prob}
The proof of Lemma~\ref{lem:success-prob} consists of two key steps: First, we argue that the algorithm terminates with success (that is, without failure) if the random choices made by the algorithm are ``consistent'' (to be defined more precisely below) with some hypothetical optimum solution. Second, we argue that consistency is maintained with sufficiently high probability in each iteration. Together with our upper bound on the number of iterations from Lemma~\ref{lem:termination}, this completes the proof of the main result,  Theorem~\ref{thm:main}.

\paragraph*{Consistency.} Informally speaking, we mean by consistency that  (i) the points in the requests are assigned to the correct (optimal) cluster and (ii) the radius $r$ in any request $(p,r)$ is justified, that is, not smaller than the distance of $p$ to its optimal center.

\begin{definition}\label{def:consistency}
    Consider a fixed hypothetical optimum solution $O=(o_1,\dots,o_k)$. We say that the current state of execution (specified by $(X,Q_1,\dots,Q_k)$) of Algorithm~\ref{alg:framework-knc} is \emph{consistent with $O$} if for any request $(p,r)\in Q_\iletter$, $\iletter\in[k]$, we have that $\delta(p,o_\kappa)\leq r$.
\end{definition}

If the current state is consistent with the optimum solution $O$, then $O$ certifies existence of solution to the cluster constraints $(Q_1,\dots,Q_k)$ currently imposed. Therefore, the following observation is straightforward.
\begin{observation}\label{obs:consistency-non-failure}
    If the state of the algorithm is consistent with $O$ before executing line~\ref{alg:update-center} in any iteration, then the algorithm does not fail during this iteration.
\end{observation}

\paragraph*{Probability of Maintaining Consistency.} If the state of execution is consistent with $O$ at the beginning of some iteration, then it remains consistent under the following two conditions. First, the point $p$ sampled in this iteration is (randomly) assigned to the correct cluster. Second, the distance of $p$ to the current candidate solution is sufficiently larger than its distance to $O$, thereby justifying the request made in line~\ref{alg:augment-cluster}. This second condition motivates the following definition.

\begin{definition}\label{def:witness}
    Given a candidate solution $X$ with $f(\vcd(P,X))>(1+\epsilon)\opt$, a point $p\in P$ is called an \emph{$\epsilon$-witness} if $\delta(p,X)>(1+\epsilon)\delta(p,O)$.
\end{definition}

The following lemma implies that the request made in any iteration for the sampled point is justified with probability $\Omega(\epsilon)$. It is a key part of our analysis as it links the specific way of (i) computing the initial upper bounds and (ii) sampling a witness based on these upper bounds. It is ultimately this interplay that allows us to bound the aspect ratio of the radii in the requests for a particular cluster and therefore the overall number of requests per cluster in terms of $k$ and $\epsilon$.
\begin{lemma}\label{lem:sample-witness}Consider a fixed iteration of the while loop of Algorithm~\ref{alg:framework-knc} and let $X$ be the candidate solution at the beginning of this iteration. The point sampled in line~\ref{alg:sample-witness} is then an $\nicefrac{\epsilon}{3}$-witness for~$X$ with probability $\Omega(\epsilon)$. In particular, the set $A$ computed in line ~\ref{alg:admissable-set} is not empty.
\end{lemma}
\begin{proof}
    For any subset $S\subseteq P$ of points let $C_S=\sum_{p\in S}w(p)\delta(p,X)$ denote the \emph{contribution of $S$} towards $\vcw^\intercal\vcd(P,X)=C_P$.
    
    Let $W\subseteq P$ be the subset of $\nicefrac{\epsilon}{3}$-witnesses of $X$. We claim that the contribution $C_W$ is at least~$\epsilon C_P/10$. Suppose for the sake of a contradiction that their contribution is less. Then, using $0<\epsilon<1$,
    \begin{align*}
        \opt & \geq \vc{w}^{\intercal}\vcd(P,O)\\ 
             & \geq \sum_{p\in P\setminus W}w(p)\delta(p,O)\\
             & \geq \frac{1}{1+\nicefrac{\epsilon}{3}}\sum_{p\in P\setminus W}w(p)\delta(p,X)\\
             & \geq \frac{1-\nicefrac{\epsilon}{10}}{1+\nicefrac{\epsilon}{3}}\sum_{p\in P}w(p)\delta(p,X)\\
             & \geq \frac{\vcw^{\intercal}\vcd(P,X)}{1+\nicefrac{\epsilon}{2}}\\
             & \geq \frac{f(\vcd(P,X))}{(1+\nicefrac{\epsilon}{2})(1+\nicefrac{\epsilon}{10})}\\
             & \geq \frac{f(\vcd(P,X))}{1+3\epsilon/4}
    \end{align*}
    which contradicts $f(\vcd(P,X))>(1+\epsilon)\opt$.

    Let $W_1,\dots,W_k$ denote the subsets of the witnesses closest to centers $x_1,\dots,x_k$ in $X$, respectively. 

     Let $H\subseteq [k]$ be the subset of clusters $\iletter\in[k]$ such that $C_{W_\iletter}\geq \epsilon C_P/(100k)$. Fix any cluster $\iletter\in H$.  Let $\{z_1,\dots,z_{\ell}\}$ be the witnesses in $W_\iletter$ in non-decreasing order by the distance $\delta(z_i,x_\iletter)$, $i\in[\ell]$ to their closest cluster center $x_\iletter$. Let $j\in[\ell]$ be the minimum index $j$ such that the contribution of the set $W_\iletter^{-}=\{z_1,\dots,z_j\}$ is at least  $C_{W_\iletter}/2$. This implies that also $C_{W_\iletter^{+}}\geq C_{W_\iletter}/2$ where $W_\iletter^{+}=\{z_j,\dots,z_{\ell}\}$. Hence $C_{W_\iletter^{-}}$ and $C_{W_\iletter^{+}}$ are both at least $\epsilon C_P/(200k)$ because $\iletter\in H$. 

     We claim that $W_\iletter^{+}\subseteq A$ where $A$ is defined as in line~\ref{alg:admissable-set} in Algorithm~\ref{alg:framework-knc}. Towards this, let $p\in W_\iletter^{+}$ be arbitrary. We prove that $u(p)\leq 1000k\delta(p,x_\iletter)/\epsilon$ and hence $p\in A$. To see this, notice that $\ball(p,2\delta(p,x_\iletter))\supseteq\ball(x_\iletter,\delta(p,x_\iletter))\supseteq W_\iletter^{-}$. On the other hand,
     \begin{displaymath}
         \frac{\epsilon\opt}{300k}\leq\frac{\epsilon f(\vcd(P,X))}{300k}\leq \frac{\epsilon C_P}{200k}\\
         \leq \sum_{q\in W_\iletter^{-}}w(q)\delta(q,x_\iletter)\\
         \leq \delta(p,x_\iletter)\sum_{q\in W_\iletter^{-}}w(q)\,.
     \end{displaymath}
     Setting $\alpha=6\delta(p,x_\iletter)$, this implies that
     \begin{equation}\label{eq:heavy-contrib}
         f(\vc{1}_{p,\nicefrac{\alpha}{3}}) \geq \vc{w}^{\intercal}\vc{1}_{p,\nicefrac{\alpha}{3}}\\
                               = \sum_{q\in\ball(p,\nicefrac{\alpha}{3})}w(q)\\
                               \geq \sum_{q\in W_\iletter^{-}}w(q)\\
                               \geq \frac{\epsilon\opt}{200k\delta(p,x_\iletter)}\\
                               = \frac{\epsilon\cdot 3\opt}{100k\alpha}\,.
     \end{equation}

     Hence $u(p)\leq 100k\alpha/\epsilon \leq 1000k\delta(p,x_\iletter)/\epsilon$ as claimed. This completes the proof of the claim that $W_\iletter^{+}\subseteq A$ for any $\iletter\in H$.
     
     As shown above, $\sum_{\iletter\in[k]}C_{W_\iletter}=C_W\geq\epsilon C_P/10$. By definition of $H$, we have $\sum_{\iletter\in [k]\setminus H}C_{W_\iletter}\leq \epsilon C_P/100$. Hence $\sum_{\iletter\in H}C_{W_\iletter}\geq \epsilon C_P/20$. Also, by the arguments above, 
     \begin{displaymath}
     C_{A\cap W}\geq \sum_{\iletter\in H}C_{W_\iletter^+}\geq \sum_{\iletter\in H}\frac{C_{W_\iletter}}{2}\geq \frac{\epsilon C_P}{40}\geq \frac{\epsilon C_A}{40}\,.
     \end{displaymath}
     Since we sample a point $p$ from $A$ with probability proportional to its contribution $C_{\{p\}}$, we sample a witness in each iteration with probability at least $\nicefrac{\epsilon}{40}$.

     Notice that $C_P\geq f(\vcd(P,X))/2>0$. The left-hand side of Equation~\ref{eq:heavy-contrib} must therefore be positive. This implies that $A$ is not empty.
\end{proof}

\paragraph*{Overall Success Probability.} We are now ready to prove Lemma~\ref{lem:success-prob}, thereby completing the proof of the main theorem~\ref{thm:main}. We establish that the state of execution is consistent before entering the while loop in Algorithm~\ref{alg:framework-knc}. The proof is completed by combining the upper bound on the number of iterations (Lemma~\ref{lem:termination}) with the lower bound on the probability of maintaining consistence (Lemma~\ref{lem:sample-witness}).

\begin{proof}[Proof of Lemma~\ref{lem:success-prob}]
    Let $p^{(1)},\dots,p^{(k')}$ be the points marked in line~\ref{step:mark-points} of Algorithm~\ref{alg:framework-knc}. By Lemma~\ref{lem:anchors}, each $\ball(p^{(\iletter)},u(p^{(\iletter)}))$, $\iletter\in [k']$ contains a point from $O$. By construction, these balls are moreover pairwise disjoint. Hence, by relabeling the optimum centers $O=(o_1,\dots,o_k)$, we can assume that $\delta(p^{(\iletter)},o_\iletter)\leq u(p^{(\iletter)})$ for each marked point $p$ where $\iletter\in [k']$ is the index of the cluster. Therefore the state of execution of the algorithm is consistent with $O$ just before the first execution of the while loop (lines~\ref{step:loop}--\ref{step:loop-end}).
    Assume now that the state is consistent with~$O$ at the beginning of an iteration of the while loop. By Lemma~\ref{lem:sample-witness}, we sample an $\nicefrac{\epsilon}{3}$-witness $p$ in this iteration with probability $\Omega(\epsilon)$. In this event, the request $(p,r)$ added has radius $r=\delta(p,X)/(1+\nicefrac{\epsilon}{3})\geq \delta(p,O)$. If additionally the cluster index $\iletter\in[k]$ picked at random is the same as the one in $O$---which happens with probability $\Omega(\nicefrac{1}{k})$---then the state remains consistent with $O$. In this event, the recomputation of the center in line~\ref{alg:update-center} does not fail. By Lemma~\ref{lem:termination}, the algorithm terminates after at most $\OO\left(\frac{k(\log\nicefrac{k}{\epsilon})\lambda(\nicefrac{\epsilon}{10})}{\epsilon}\right)$ many iterations. Since in any iteration it does not fail with probability $\Omega(\nicefrac{\epsilon}{k})$, it succeeds overall with probability $\exp{\left(-\widetilde{\OO}\left(\frac{k\lambda(\nicefrac{\epsilon}{10})}{\epsilon}\right)\right)}$.
\end{proof}

\section{\ScatterDim Bounds}\label{sec:applications}

This section is devoted to bounding the \scatterdim in various classes of metrics, proving Theorems~\ref{thm:doubling-metric}, \ref{thm:treewidth}, and \ref{thm:planar} from the Introduction.

\subsection{Bounded Doubling Dimension}\label{sec:bounded-doubling}
In this section, we show the 
upper bound of the \scatterdim of any metric space of doubling dimension $d$, proving Theorem~\ref{thm:doubling-metric}.

\paragraph{Scatter Dimension and Packing.} 
Given metric $(X,\delta)$, an $\epsilon$-packing of this metric is a subset of points $X' \subseteq X$ such that $\delta(i,j) \geq \epsilon$ for all $i,j \in X'$. This is a standard notion in the theory of metric spaces.  
We first observe the following connection between our \scattering  and $\epsilon$-packing.

\begin{observation}\label{obs:scatttering-sequence}
Let $(x_1,p_1), \ldots, (x_{\ell},p_{\ell})$ be an \scattering  in a metric space $(P,F,\delta)$. Then, the centers $X=\{x_2,\ldots, x_{\ell}\}$ is an $\epsilon$-packing in metric $(P \cup F, \delta)$ and $X$ is contained in a unit ball. 
\end{observation}
\begin{proof}
 
From the definition of \scattering, we have $\delta(p_1,x_j) \leq 1$ for each $j \in \{2, \ldots, \ell\}$.
Also, for any $i,j \in \{1, \ldots, \ell\}$, since $\delta(x_j,p_j) > 1 + \epsilon$ and $\delta(p_j,x_i)\leq 1$ for $j<i$, by the triangle inequality $\delta(x_i,x_j) > \epsilon$.
If $i \neq 1$, then both $\delta(p_{i-1},x_i)$ and $\delta(p_{i-1},x_j) \leq 1$. Therefore, by the triangle inequality, $\delta(x_i,x_j) \leq \delta(x_i,p_{i-1}) + \delta(p_{i-1},x_j) \leq 2$.
\end{proof}

\begin{corollary}
The size of $\epsilon$-packing of a unit ball in metric $M$ is at most the $\epsilon$-scatter dimension minus one. 
\label{cor:packing vs scatter} 
\end{corollary}

It is a well-known fact that $\epsilon$-packing of any metric of doubling dimension $d$ has size at most $\OO((1/\epsilon)^d)$. Combining this with Observation~\ref{obs:scatttering-sequence} yields Theorem~\ref{thm:doubling-metric}.

\paragraph{Remark:} We note that the converse of Corollary~\ref{cor:packing vs scatter} is false even in a very simple graph metric such as a star. In an $n$-node star rooted at $r$, a unit ball $\ball(r,1)$ includes the whole graph. There exists an $\epsilon$-packing of size $(n-1)$ by choosing the non-root nodes. However, any \scattering has length at most $2$.

\subsection{Bounded Treewidth Graphs}\label{sec:graph-metrics}
In this section we show that any graph of treewidth $\tw$ has \scatterdim $\tw^{(1/\epsilon)^{\OO(\tw)}}$. That is, we prove Theorem~\ref{thm:treewidth} for the bounded treewidth graph metric. We later show that the bound for planar graphs can be derived via an embedding result of~\cite{friggstad2019local}. For convenience, we abbreviate $\ball_{\delta_G}(r,\gamma)$ by $\ball_G(r,\gamma)$.

\subsubsection{Treewidth and Spiders}\label{sec:tw}

Our proof relies on the notion of spiders, whose existence can serve  as a ``witness'' to the fact that the treewidth of a graph $G$ is high. 
Given an edge-weighted graph $G$, $X \subseteq V(G)$ and $\gamma \in (0,1)$, 
a {\em $\gamma$-\spider\ on $X$} is a set $S=\ball_G(r,\gamma)$ for some $r \in V(G)$ such that 
there are $|X|$ paths from $S$ to $X$ that are vertex-disjoint except for in $S$. 
We say that a set $S$ is a {\em spider} on $X$ if it is a $\gamma$-spider for some $\gamma$. 
See Figure~\ref{fig:spider} for illustration. 

\begin{figure}[t]
    \centering
    \includegraphics[width=0.5\textwidth]{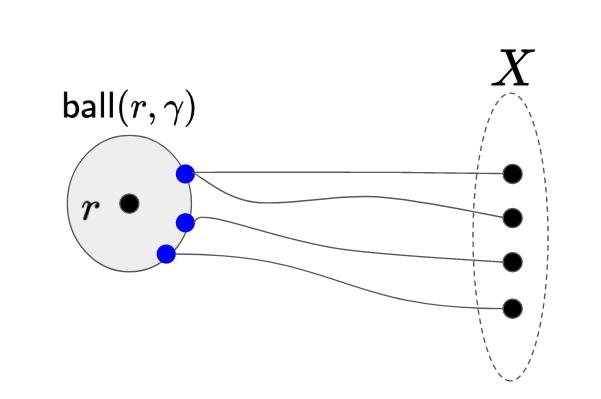}
    \caption{A spider $S= \ball_G(r,\gamma)$ on $X$. Paths connecting $X$ to $r$ are disjoint, except for nodes in $S$.}
    \label{fig:spider}
\end{figure}

Observe that if $S$ is a $\gamma$-spider on X, then for any $X' \subseteq X$, $S$ is also a $\gamma$-spider on $X'$. 
The following lemma is key to our result, roughly showing that the existence of a large number of spiders implies that the treewidth of $G$ is large. 

\begin{lemma}\label{lem:bound-spider}
Let $G$ be a graph, $k$ be an integer and $X \subseteq V(G): |X| > 3k$. If there is a family ${\mathcal S}$ of $k+1$ pairwise disjoint spiders on $X$, then the treewidth of $G$ is larger than $k$. 
\end{lemma}
\begin{proof}
Assume otherwise that the treewidth is at most $k$. Then, there exists a balanced separator $A \subseteq V(G)$ such that $|A| = k$ and  a partition of $G-A = V_1 \uplus V_2$ such that $E(V_1,V_2) =\emptyset$, $|V_1 \cap X |, |V_2 \cap X| \geq |X|/3$~\cite{DBLP:books/sp/CyganFKLMPPS15} (Folklore). 

We claim that  each spider $S \in \mathcal{S}$ must contain a vertex in the separator, i.e., $S \cap A \neq \emptyset$. For the sake of contradiction, say there exists $S \in \mathcal{S}$ such that $S \subseteq V_i$ for some $i \in \{1,2\}$. Without loss of generality, let $S \subseteq V_1$. Recall that $S$ is a spider on $X$. Hence there are $|X|$ many internally-vertex disjoint paths from $S$ to distinct vertices of $X$. 
Since $|V_2 \cap X| \geq |X| /3 > k$, there are at least $k+1$ internally vertex-disjoint paths from $S$ ($\subseteq V_1$) to $V_1 \cap X$. Since $A$ is a $(V_1,V_2)$-separator, all these vertex-disjoint paths pass through $A$. Thus, $|A|$ is at least the number of these paths ($k +1$), which is a contradiction. We conclude that each spider $S \in \mathcal{S}$ intersects $A$. 

Since $\mathcal{S}$ is a family of pairwise vertex-disjoint spiders, we conclude that $|A| \geq k+1$, a contradiction.
\end{proof}

\subsubsection{Iteratively Finding Spiders}

Our main result in this section is encapsulated in the following theorem. 

\begin{theorem}
\label{thm:structural thm for tw} 
If there is an \scattering of length at least $(\OO(k/\epsilon))^{(4/\epsilon)^{k+1}}$ in $G$, then graph $G$ contains a family of $k+1$ disjoint spiders on vertex set of size greater than $3k$.  
\end{theorem}

Combining the above with Lemma~\ref{lem:bound-spider}, we can deduce that the  length of any \scattering is at most $\tw^{(1/\epsilon)^{\OO(\tw)}}$ as desired. 
We spend the rest of this section proving the theorem. 
Given \scattering $\sigma$, we say that the $\epsilon$-packing $X= X(\sigma)$, given by Observation~\ref{obs:scatttering-sequence}, is a canonical packing of $\sigma$. 

\begin{lemma}
\label{lem: tw iteration} 
Let $\sigma$ be an \scattering of length $\ell$ in $G\subseteq \ball_G(r,1)$ and $X= X(\sigma)$  its canonical $\epsilon$-packing. Then, there exist 
\begin{itemize}
\item a spider $S= \ball_G(r,\epsilon/3)$ on $X' \subseteq X: |X'| \geq c_0 \epsilon  \cdot (\ell/2)^{\epsilon/3}$ for some constant $c_0$ and 

\item a graph $G'$  such that $S \cap V(G') = \emptyset$ and an \scattering $\sigma'$ that is a subsequence of $\sigma$ such that $X(\sigma') = X'$.  
\end{itemize}
\end{lemma}

Before proving this lemma, we show how it implies Theorem~\ref{thm:structural thm for tw}. 
Let $G_0 = G$ contain a \scattering  $\sigma_0$ of length at least  $\ell_0 =  (\frac{k}{c_0 \epsilon})^{(4/\epsilon)^{k+1}}$ and $X_0 = X(\sigma_0)$. The lemma allows us to find a spider $S_1$ on $X_1$ of size $c_0 \epsilon \cdot \ell_0^{\epsilon/3} \geq (\frac{k}{c_0 \epsilon})^{(4/\epsilon)^{k}} = \ell_1$ for sufficiently small $\epsilon$. Moreover, we have the graph $G_1$ that is disjoint with $S_1$ and \scattering that is a subsequence $\sigma_1$ of length $\ell_1$. 
Since $(G_1,X_1,\sigma_1)$ satisfies the preconditions of Lemma~\ref{lem: tw iteration}, we can apply it to obtain $(G_2,X_2,\sigma_2)$ and so on. 
More formally, starting from $(G_i, \sigma_i, X_i)$, we apply Lemma~\ref{lem: tw iteration} to obtain $(G_{i+1}, \sigma_{i+1}, X_{i+1})$. We maintain the following invariant: The length of the sequence $\sigma_i$ satisfies $\ell_i = |X_i| \geq (\frac{k}{c_0  \epsilon})^{(4/\epsilon)^{k+1-i}}$. This allows us to find disjoint spiders $S_1, S_2, \ldots, S_{k+1}$ on $X_{k+1}: |X_{k+1}| > 3k$ as desired. 

\subsubsection{Proof of Lemma~\ref{lem: tw iteration}}  

Let $G$ be contained in the unit ball $\ball(r,1)$. 
The proof has two steps. In the first step, we find a spider $S$ on a subset $X''\subseteq X$ of relatively large size. In the second step, we show the graph $G'$ obtained by removing $S$ from $G$ still contains a large subsequence $\sigma'$ of $\sigma$ whose canonical packing is a subset $X'$ of $X''$ that has desired cardinality. 

\paragraph{First step:}
Let $T$ be a shortest path tree from $r$ to $X$ (recall that $|X| = \ell$), so vertices in $X$ appear at the leafs of this tree.
We construct an ``auxiliary'' tree  $\widehat{T}$ on subset $\widehat{V} \subseteq V(T)$ from $T$ inductively as follows. 
Let $B_r = \ball_T(r,\epsilon/3)$. Remove $B_r$ from $T$ to obtain subtrees $T_1,\ldots, T_q$ with roots $r_1,\ldots, r_q$. For each $i \in [q]$, let $X_i \subseteq X$ be the descendants of $r_i$ in $T_i$ that are in $X$. Since vertices in $X$ are at the leaf, we have that $X_i \neq \emptyset$.  
We inductively perform this process on the instances $(T_1,X_q),\ldots, (T_q,X_q)$ to obtain the auxiliary subtrees $\widehat{T}_i$ for $(T_i, X_i)$. 
Now create $\widehat{T}$ by connecting $r$ to $r_1,\ldots, r_q$ in (making them direct children of $r$). See Figure~\ref{fig:construction-aux-trees}. 

\begin{figure}[t]
    \centering
    \includegraphics[width=0.8\textwidth]{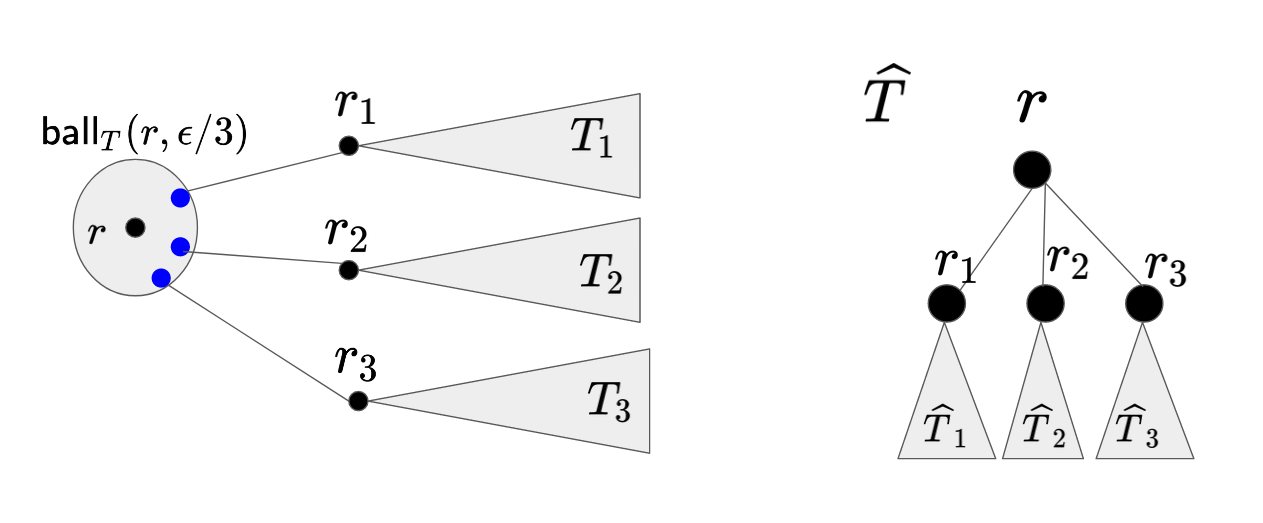}
    \caption{A recursive construction of tree $\widehat{T}$.}
    \label{fig:construction-aux-trees}
\end{figure}

For each $v \in V(\widehat{T})$, denote by $T_v$ the subtree of $T$ rooted at $v$ and $B_v$ the ball $\ball_{T_v}(v,\epsilon/3)$ constructed by the recursive procedure. 
Observe that $\bigcup_{v \in V(\widehat{T})}B_v \supseteq X$ and that the depth of $\widehat{T}$ is at most $(3/\epsilon)$ (since $\delta(r,x) \leq 1$ for all $x \in X$ and each recursion reduces the root-to-leaf distance by $\epsilon/3$.)  

\begin{claim}
There must be a vertex $r' \in V(\widehat{T})$ such that $r'$ has at least $D=(\ell/2)^{\epsilon/3}$ children in $\widehat{T}$.       
\end{claim}
\begin{proof}
Assume that the number of children is less than $D$ for every vertex in $\widehat{T}$.
Then the total number of vertices in $\widehat{T}$ is less than $2D^{3/\epsilon}$. For each such vertex $v \in V(\widehat{T})$, we have $|B_v \cap X| \leq 1$ (since $X$ is an $\epsilon$-packing while the diameter of $B_v$ is at most $2\epsilon/3$). Therefore, $\ell =|X| \leq \bigcup_{v} |B_v \cap X| \leq 2D^{3/\epsilon}$. This would imply that $D \geq (\ell/2)^{\epsilon/3}$.  
\end{proof}

Let $v$ be the node in $\widehat{V}$ closest to the root in $\widehat{T}$ such that there are at least $D$ children (breaking ties arbitrarily).
This means that (in the process of creating $\widehat{T}$) removing $B_v=\ball_{T_v}(v,\epsilon/3)$ gives us at least $D$ subtrees $T_1,\ldots, T_D$ where each such tree contains (arbitrarily chosen) $x_i \in X_i$ as a descendant in $T$. Notice that $S$ is a spider on $X'' = \{x_1, x_2,\ldots, x_D\}$.  

\paragraph{Second step:} 
Let $\sigma''$ be the subsequence of $\sigma$ whose canonical $\epsilon$-packing is $X''$, that is, $X(\sigma'') = X''$. Recall that $|X''| \geq (\ell/2)^{\epsilon/3}$. Denote the spider $S$ by $S = \ball_G(s,\epsilon/3)$. 
In this second step, we show that $G' = G \setminus S$ still contains a long \scattering $\sigma'$ which is a subsequence of $\sigma''$ such that $X'=X(\sigma') \subseteq X''$ has the desirable length. 

By construction, we have that $X'' \subseteq \ball_G(s,1)$. 
We partition vertices in $X''$ into at most $3/\epsilon$ subsets based on the distances to $s$ as follows: For $i =1,\ldots, 3/\epsilon$, let $X''_i = \{x \in X'': \delta_G(s,x) \in (i\epsilon/3,(i+1)\epsilon/3]\}$.  Define $X'$ to be the set $X''_i$ that has maximum cardinality, so $|X'| \geq \frac{\epsilon}{3} \cdot  (\ell/2)^{\epsilon/3}$. Note that for all $u,u' \in X'$, we have $|\delta(s,u) - \delta(s,u')| < \epsilon/3$. 
The following claim asserts that each point in $S$ is roughly of the same distance from every point in $X'$ (see Figure~\ref{fig:decomp} for illustration). 

\begin{figure}
    \centering
    \includegraphics[width=0.5\textwidth]{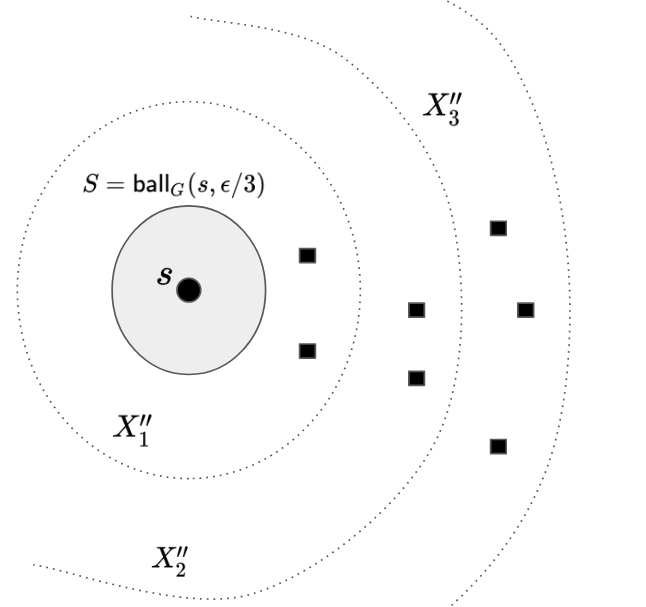}
    \caption{The partition of $X''$ into $\{X''_i\}$ based on their distance from the spider $S$. Rectangular points are the points in $X''$.}
    \label{fig:decomp}
\end{figure}

\begin{claim}
Let $v \in S$ and $u,u' \in X'$. Then $|\delta(v,u) - \delta(v,u')| < \epsilon$. 
\end{claim}
\begin{proof}
Assume w.l.o.g. that $\delta(v,u) \leq \delta(v,u')$. By the triangle inequality, $\delta(v,u') \leq \delta(v,s) + \delta(s,u') \leq \delta(v,s) + \delta(s,u) + \epsilon/3 \leq \delta(s,u) + 2\epsilon/3$. Applying triangle inequality again, we get $\delta(s,u) \leq \delta(s,v) + \delta(v,u)$, and hence the desired bound.      
\end{proof}

The following claim will finish the proof.

\begin{claim}
\label{claim: same distances} 
Let $\sigma'$ be the subsequence of $\sigma''$ whose canonical packing is $X'$. Then $\sigma'$ is a valid \scattering in $G'= G \setminus S$. 
\end{claim}
\begin{proof}
We abbreviate $\delta_{G'}$ simply by $\delta'$. 
Let $(x,p)$ and $(x',p')$ be two pairs in $\sigma'$ such that $(x,p)$ appears before $(x',p')$. Since $\sigma$ is scattering, we have that $\delta(x',p) \leq 1$ while $\delta(x,p), \delta(x',p'), \delta(x,p') > (1+\epsilon)$. 

Notice that the refutation properties hold for these pairs after removing $S$, i.e., $\delta'(x,p), \delta'(x',p'),\allowbreak \delta'(x,p') > (1+\epsilon)$ (the distances cannot decrease after removing vertices from a graph).
It suffices then to show that $\delta'(p,x') \leq 1$. To this end, we argue that any shortest path from $p$ to $x'$ in $G$ cannot intersect with the ball $S$. 
Assume otherwise that there exists a shortest path $Q$ from $p$ to $x'$ in $G$ that intersects with $S$ at some vertex $v \in S \cap Q$. Notice that $\delta(p,x') = \delta(p,v) + \delta(v,x')$. We will reach a contradiction by showing that $\delta(p,x) \leq (1+\epsilon)$. 
Since $\delta(p,x) \leq \delta(p,v) + \delta(v,x)$, by Claim~\ref{claim: same distances}, this is at most $\delta(p,v) + \delta(v,x') + \epsilon =\delta(p,x') + \epsilon$, which would imply that $\delta(p,x) \leq (1+\epsilon)$, contradicting to the refutation property.  
\end{proof}

\subsection{Bounding \ScatterDim via Low-Treewidth Embedding}
\label{sec:planar}

In this section, we show a (simple) connection between bounding \scatterdim and an active research area on embedding with additive distortion~\cite{fox2019embedding,filtser2022low,cohen2020light}. This connection allows us to upper bound the \scatterdim of planar graphs.

In particular, we say that (weighted) graph class ${\mathcal G}$ admits a \textit{$t$-low treewidth-diameter embedding} for function $t: {\mathbb N} \rightarrow {\mathbb N}$ if there exists a \textit{deterministic} algorithm that takes $G$ and produces weighted graph $H$ of treewidth at most $t(\eta)$ and an embedding $\phi: V(G) \rightarrow V(H)$ such that: 
\[ \delta_G(u,v) \leq \delta_H(\phi(u), \phi(v)) \leq  \delta_G(u,v) + \eta D\]
where $D$ is the diameter of $G$. 

\begin{theorem}
Let $\lambda_{\tw}(\epsilon)$ denote the  the \scatterdim of graphs of treewidth $\tw$ (from the previous section, this bound is at most doubly exponential in $\tw$). 
If graph class ${\mathcal G}$ admits a $t$-low treewidth-diameter embedding, then every metric in ${\mathcal G}$ has \scatterdim at most $\lambda_{t(\epsilon/10)}(\epsilon/3)$.      
\end{theorem}
\begin{proof}
Let $(x_1,p_1), (x_2,p_2), \ldots, (x_{\ell}, p_{\ell})$ be \scattering  in $G$. Let $\eta = \epsilon/10$. Consider an embedding $\phi$ of $G$ into $H$ such that the treewidth of $H$ is at most $t(\eta) = t(\epsilon/10)$. Notice that     
\begin{itemize}
    \item For $1 \leq i < j \leq \ell: \delta_H(\phi(x_j), \phi(p_i)) \leq 1 + 2\eta = 1+\epsilon/5$
    \item For $1 \leq i \leq \ell: \delta_H(\phi(x_i), \phi(p_i)) > 1 + \epsilon$
\end{itemize}
Consider (weighted) graph $H'$ obtained by scaling the weights of $H$ down by a factor of $(1+\epsilon/5)$. We have that $\delta_{H'}(\phi(x_j), \phi(p_i)) \leq 1$ while $\delta_{H'}(\phi(x_i), \phi(p_i)) > \frac{1+\epsilon}{1+\epsilon/5} \geq 1+ \epsilon/3$ for sufficiently small $\epsilon >0$. This implies that the embedded sequence is $(\epsilon/3)$-scattering in $H$. Therefore, from Theorem~\ref{thm:treewidth}, the length $\ell$ is upper bounded by $\lambda_{t(\epsilon/10)}(\epsilon/3)$.  
\end{proof}

Now we can use the following theorem. 

\begin{theorem}[Theorem 1.3 of \cite{fox2019embedding}]
\label{thm:embedding}
There is a polynomial-time algorithm that, given an edge-weighted planar graph and given a number $\eta > 0$, outputs an embedding of the graph into a planar graph of treewidth  $\poly(1/\eta)$ with additive error $\eta \cdot D$ where $D$ is the diameter of the input graph.
\end{theorem}

This implies, in our language, that planar graphs have low treewidth-diameter embedding.

\begin{corollary}
Planar graphs have  \scatterdim at most $\exp \left( \exp ({\sf poly}(1/\epsilon)) \right)$.  
\end{corollary}

\subsection{High-Dimensional Euclidean Space}\label{sec:high-dim-cont}

Recall, from the introduction and Sections~\ref{sec:scatter-dim}, that the \scatterdim of high-dimensional (continuous) Euclidean space  is unbounded. In this section, we show, however, that the \algscatterdim of this metric is bounded.

\conteuclscatter*
We dedicate the rest of this section to the proof of Theorem \ref{thm:scatter-cont-eucl}. 
In order to upper bound the \algscatterdim for the continuous Euclidean space, it suffices to show that there exists an algorithm $\cC$ such that the $(\cC, \epsilon)$-scattering dimension in the Euclidean space is bounded. We use an algorithm by Kumar and Yildirim~\cite{KY09} as \pc algorithm for the high-dimensional Euclidean space. They study the \pc problem in the language of \textsc{Weighted Euclidean $1$-Center}. They provide a \pc algorithm based on a convex optimization formulation which efficiently (and approximately) solves the \pc problem in continuous Euclidean setting for weights with bounded aspect ratio. Let $\cC_{\textnormal{KY}}$ denote this algorithm. The following lemma is adapted from Kumar and Yildirim's work into our terminology (see Lemma 4.2 of \cite{KY09}).

\begin{lemma}\label{lem:ky}
Given an instance $(P, F, \delta)$ of \pc in high-dimensional Euclidean space, associated radii $r(p)$
to each $p \in P$, and $\epsilon \in (0, 1)$, the length of any \algscattering[$(\cC_{\textnormal{KY}}, \epsilon)$]
is at most $\bigO{\nicefrac{\tau}{\epsilon^2}}$ where $\tau \geq 1$ is the squared ratio of the largest radius in the requests to the smallest.
\end{lemma}

Note that for a constant $\tau$, Lemma \ref{lem:ky} yields the proof of the theorem. 
To complete the proof, we show that by increasing the length of the \scattering by a multiplicative factor of $\bigO{\log \nicefrac{1}{\epsilon}}$, we can assume that $\tau$ is $\bigO{\nicefrac{1}{\epsilon^2}}$. 

\paragraph*{Aspect-Ratio Condition.} The following lemma provides a sufficient condition for bounded \algscatterdim that facilitates the design of a \pc algorithm for bounding the \algscatterdim . In particular, this condition is key to bound the \algscatterdim of high-dimensional continuous Euclidean spaces. It can be seen as a strenghtened converse of the aspect-ratio lemma~\ref{lem:aspect-ratio}.
\begin{lemma}[Aspect-Ratio Condition]\label{lem:critertion-algscattercomp}
  Let $\cM$ be a class of metric spaces with \pc algorithm $\cC_{\cM}$ and let $\epsilon\in (0,1)$. If any $(\cC_{\cM}, \epsilon)$-scattering $(x_1, p_1, r_1), \ldots, (x_\ell, p_\ell, r_\ell)$ with $r_i \in [\nicefrac{\epsilon}{12}, 1]$, $i \in [\ell]$ contains at most $\lambda(\epsilon)$ triples with the same radius, then the \algscatterdim of $\cM$ bounded by $\bigO{\lambda(\epsilon)\log\nicefrac{1}{\epsilon}}$.
\end{lemma}

 To prove Lemma~\ref{lem:critertion-algscattercomp} , we assume that we are given a \pc algorithm $\cC_{\cM}$ as stated. We claim that the following \pc algorithm, 
which invokes $\cC_{\cM}$ as a sub-routine, yields \algscatterdim $\OO(\lambda(\epsilon)\log\nicefrac{1}{\epsilon})$ according to the condition of Definition~\ref{def:algscattercomp}.

\begin{algorithm}[H]
    \caption{\pc algorithm realizing lemma~\ref{lem:critertion-algscattercomp}.}\label{alg:aspect-ratio-alg-scatter-comp}
  \KwData{Metric space $M = (P,F,\delta)\in\cM$, requests $Q = (p_1, r_1), \ldots (p_\ell, r_\ell)$ with $p_i \in P$ and $r_i \in \bR^{+}$ for $i \in [\ell]$, error parameter $\eta>0$, \pc algorithm $\cC_{\cM}$ as in Lemma~\ref{lem:critertion-algscattercomp}}
  \KwResult{center $x\in F$ such that $\delta(p_i,x)\leq (1+\eta)r_i$ for all $i \in [\ell]$ or ``fail''}
  $\rho\gets\min\{\,2^{-j}\mid j\in\bN_{0}\textnormal{ and }\min_{i \in [\ell]}r_i \leq 2^{-j}\,\}$\;
  $Q'\gets\{\,(p_i, r_i) \mid i \in [\ell], \nicefrac{\eta}{3}\cdot r_i\leq \rho \,\}$\;
  $x\gets\cC_{\cM}(M, Q', \eta)$\;  
  \ForEach{$i\in [\ell]$}{
    \lIf{$\delta(p_i,x)>(1+\eta)r_i$}{
      ``fail''
    }
   }
  \Return{$x$}\;
\end{algorithm}
The following definition formulates a condition for two requests $(p, r), (p', r')\in Q$ under which it suffices to satisfy $(p, r)$ in order satisfy $(p', r')$ as well.
\begin{definition}\label{def:implication}
Let $\eta\in (0,1)$ and $(p, r), (p', r') \in Q$. We say that $p$ \emph{$\eta$-implies} $p'$ if $\ball(p,(1+\eta)r)\subseteq\ball(p',(1+\eta)r')$.
\end{definition}

\begin{lemma}\label{lem:cluster-implication}
Let $(p, r), (p', r') \in Q$ be two requests such that $r\leq \nicefrac{\eta}{3}\cdot r'$ for some $\eta\in (0,1)$. If there is some center in $F$ satisfying both requests then $p$ $\eta$-implies $p$.
\end{lemma}
\begin{proof}
    Let $o$ be the center satisfying both requests and let $x\in \ball(p,(1+\eta)r(p))$. By triangle inequality
\begin{align*}
    \delta(p',x) & \leq \delta(p',o)+\delta(o,p)+\delta(p,x)\\
    &  \leq r'+r+(1+\eta)r\\
    & \leq (1+\eta)r'\qedhere
\end{align*}    
\end{proof}

\begin{lemma}
    Algorithm~\ref{alg:aspect-ratio-alg-scatter-comp} is a \pc algorithm.
\end{lemma}

\begin{proof}
    Assume there is some $o\in F$ such that $\delta(p_i,o)\leq r_i$ for all $i \in [\ell]$. We want to show that our algorithm does not fail and outputs a center $x$ such that $\delta(p_i, x) \leq (1+\eta) r_i$ for any $i \in [\ell]$.
    
    Let $Q',\rho,x$ be defined as in Algorithm~\ref{alg:aspect-ratio-alg-scatter-comp}.
    Consider any $(p_i, r_i) \in Q$. We distinguish two cases. First assume that $(p_i, r_i) \in Q'$. Assuming that~$\cC_{\cM}$ is a correct \pc algorithm, we have that $\delta(p_i, x) \leq (1 + \eta) r_i$. On the other hand, if $(p_i, r_i) \in Q \setminus Q'$, then there is some $(p_j, r_j) \in Q'$ such that $r_j \leq \rho < \nicefrac{\eta}{3} \cdot r_i$. Hence $p_j$ $\eta$-implies $p_i$ by Lemma~\ref{lem:cluster-implication}. As argued above,
    $x \in \ball(p_j, (1+\eta) r_j)$ and hence $x \in \ball(p_i, (1+\eta) r_i)$ by Definition~\ref{def:implication}.
\end{proof}

We conclude the proof of Lemma~\ref{lem:critertion-algscattercomp} by proving the following lemma.
\begin{lemma}
    Let $\cC$ denote Algorithm~\ref{alg:aspect-ratio-alg-scatter-comp}. Any $(\cC, \epsilon)$-scattering has then $\OO(\lambda(\epsilon)\log\nicefrac{1}{\epsilon})$ triples with the same radius.
\end{lemma}
\begin{proof}
        Let $Q', \rho, x$ be defined as in Algorithm~\ref{alg:aspect-ratio-alg-scatter-comp}. 
        Consider an arbitrary $(\cC, \epsilon)$-scattering $(x_1, p_1, r_1), \allowbreak \ldots, (x_\ell, p_\ell, r_\ell)$ and let $Q$ denote the sequence of its requests, $(p_1, r_1), \ldots, (p_\ell, ,r_\ell)$. Notice that, hypothetically, if we were to run the algorithm on all prefixes of $Q$ by increasing length, the value of $\rho$ would be monotonically decreasing over the sequence. We sub-divide the sequence into \emph{phases}, which are maximal (contiguous) sub-sequences in which the value $\rho$ does not change.
        Fix some phase. Notice that the set $Q'$ would be inclusion-wise increasing during this phase because $\rho$ remains unchanged. If $(p_i, r_i)$ is an arbitrary request added to $Q'$ at some point during the phase, then $\nicefrac{\rho}{2} \leq r_i \leq \nicefrac{3\rho}{\eta}$. Re-scaling distances by factor $\nicefrac{\eta}{(3\rho)}$ and using $\eta=\nicefrac{\epsilon}{2}$ shows that all requested radii during this phase lie in the range $[\nicefrac{\epsilon}{12},1]$. Hence, by the assumption on sub-routine~$\cC_\cM$ made in Lemma~\ref{lem:critertion-algscattercomp}, the scattering has at most $\lambda(\epsilon)$ many triples per radius value~$r_i$ during this phase.
        
        We complete the proof by noting that there are at most $\log_2(\nicefrac{12}{\epsilon})$ many phases for any fixed radius value $r_i$, 
        during which request with this radius are added to $Q'$.
\end{proof}

Lemmas \ref{lem:ky} and \ref{lem:critertion-algscattercomp} together with the observation that $\tau = \left(\nicefrac{12}{\epsilon}\right)^2$ give the proof of Theorem~\ref{thm:scatter-cont-eucl}.

\section{Conclusions and Open Problems}

We present a unified view on efficient parameterized approximation schemes that applies to large variety of clustering objectives and metric spaces. 
From complexity theoretic perspective, this implies  rather surprising collapses of approximability of symmetric and asymmetric norm clustering problems (in the regime of \textbf{P}, their approximabilities are substantially different, and yet they collapse in \textbf{FPT}).

There are rooms for future open problems in two directions that can be pursued independently. First, can we characterize the class of metric spaces with bounded scatter dimension? For example, do minor-free graphs have bounded scatter dimension? 
This is a purely structural question whose resolution immediately yields an EPAS (through our framework). Since the bounded treewidth graphs play an important role in our approach (through the lens of low treewidth embedding), it would be interesting to pinpoint the exact bound on their \scatterdim; a particularly interesting concrete question is whether the bound can be brought down to singly exponential. 

The second direction concerns extensions of our framework. 
Some clustering objectives are still missing from our framework. For instance, what about clustering with outliers~\cite{KLS18, BOR21, DZ22} (in which case the cost function $f$ is instead an anti-norm)? 
Even more conceptually, our current algorithm is oblivious to the structure of the input metric, but our theorem can only talk about whether an EPAS can be obtained. Is it possible for such a framework to give approximation factors in all spectrums (e.g., $(3+o(1))$-approximation for \kc if a general, unstructured metric space is given as input)? 
The ``dream result'' could look something like $(\gamma({\mathcal M})+o(1))$ FPT approximation algorithm where $\gamma({\mathcal M})$ is the FPT-approximability of metric class ${\mathcal M}$.

Last, but not least, the scope of this paper is to handle multiple clustering objectives and metric spaces. Many clustering problems additionally impose restrictions on how points in $P$ can be assigned to open centers in $X$, e.g., capacity~\cite{ABMM019, LI16, VC20}, different notions of fairness~\cite{BCFN19, CKLV17, CFLM19} and diversity constraints~\cite{LYZ10, TOG21,10.1145/3534678.3539487}; in such case, our framework does not apply.  Extending our framework to handle such constraints (or proving that EPASes do not exist when such constraints are enforced) is an interesting direction.

\section*{Acknowledgements}

Fateme Abbasi, Sandip Banerjee and Jaros\l{}aw Byrka are supported by Polish National Science Centre (NCN) Grant 2020/39/B/ST6/01641. Parinya Chalermsook is supported by European Research Council (ERC) under the European Union’s Horizon 2020 research and innovation programme (grant agreement No 759557). D\'{a}niel Marx is supported by the European Research Council (ERC) consolidator grant No.~725978 SYSTEMATICGRAPH.

\bibliographystyle{alpha}
\bibliography{references}

\end{document}